\documentclass[journal,subeqn,twocolumn]{IEEEtran}
\usepackage{amsfonts,amssymb,amsmath,amsthm}
\usepackage{calrsfs}
\usepackage{latexsym}
\usepackage{epsfig}
\usepackage{epstopdf}
\usepackage[font=footnotesize]{subfig}
\usepackage{cite}
\usepackage{graphicx}
\usepackage{subeqn}
\usepackage{xcolor}
\usepackage{algorithmicx, algpseudocode}

\renewcommand{\tabcolsep}{3pt}

\newtheorem{prob-statement}{Problem}
\newtheorem{defn}{Definition}
\newtheorem{lemma}{Lemma}

\newtheorem{thrm}{Theorem}

\DeclareMathOperator*{\argmax}{arg\,max}

\DeclareMathOperator*{\st}{\text{subject to}\,}
\DeclareMathOperator*{\maximize}{\text{maximize}\,}

\newcommand{\nchoosek}[2]{\left(\!\! \begin{array}{c} #1 \\ #2 \end{array} \!\!\right)}

\DeclareMathAlphabet{\mathpzc}{OT1}{pzc}{m}{it}

\newcommand{\ignore}[2]{\hspace{0in}#2}

\newcounter{MYtempeqncnt}

\newenvironment{redtext}{\par\color{black}}{\par}
\newenvironment{magentatext}{\par\color{black}}{\par}
\newenvironment{bluetext}{\par\color{black}}{\par}

\begin{document}
	\IEEEoverridecommandlockouts
\title{Optimal Spectrum Auction Design with Two-Dimensional Truthful Revelations under Uncertain Spectrum Availability}
        \author{
        \IEEEauthorblockN{V. Sriram Siddhardh Nadendla\IEEEauthorrefmark{1}, Swastik Brahma\IEEEauthorrefmark{1} and Pramod K. Varshney\IEEEauthorrefmark{1}}
        \\
        \IEEEauthorblockA{\IEEEauthorrefmark{1}Department of EECS, Syracuse University, Syracuse, NY 13244\\ Email: \{vnadendl, skbrahma, varshney\}@syr.edu}
        }

        \maketitle

\begin{abstract}
In this paper, we propose a novel sealed-bid auction framework to address the problem of dynamic spectrum allocation in cognitive radio (CR) networks. We design an optimal auction mechanism that maximizes the moderator's expected utility, when the \emph{spectrum is not available with certainty}. We assume that the moderator employs \emph{collaborative spectrum sensing} in order to make a reliable inference about spectrum availability. Due to the presence of a collision cost whenever the moderator makes an erroneous inference, and a sensing cost at each CR, we investigate feasibility conditions that guarantee a non-negative utility at the moderator. We present tight theoretical-bounds on instantaneous network throughput and also show that our algorithm provides maximum throughput if the CRs have i.i.d. valuations. Since the moderator fuses CRs' sensing decisions to obtain a global inference regarding spectrum availability, we propose a novel strategy-proof fusion rule that encourages the CRs to simultaneously reveal truthful sensing decisions, along with truthful valuations to the moderator. Numerical examples are also presented to provide insights into the performance of the proposed auction under different scenarios.

\end{abstract}


\section{Introduction}

Over the past decade, spectrum scarcity has driven several researchers to pursue the problem of dynamic spectrum access (DSA) \cite{Zhao2007} for cognitive-radios (CRs). Federal Communications Commission (FCC) relaxed some of its license protocols, after identifying the inefficient usage of spectrum by the PUs (licensed users), and encouraged the use of licensed spectrum by unlicensed secondary users or CRs, without interfering with the PUs. In order to implement the solutions offered by DSA, intelligent CRs were proposed by Mitola in \cite{MitolaPhD2000} which learn and adapt to the dynamically changing environment (spectrum). In the problem of dynamic spectrum allocation, multiple CRs compete for the same primary user’s (PU) channel, whenever it is not being used. In such a case, there is a need for a mechanism that allocates the spectrum among the CRs in the network. In this paper, we focus our attention on how auction mechanisms can resolve the problem of spectrum allocation among multiple CRs.


FCC proposed two important variants of the problem of dynamic spectrum allocation (for example, in the case of Unlicensed National Information Infrastructure (UNII) devices \cite{FCC2013}), namely \emph{deterministic} and \emph{nondeterministic} spectrum allocations. In \emph{deterministic} spectrum allocation, the moderator (base-station that allocates the spectrum) has complete knowledge about spectrum availability, while allocating the spectrum to the participating CRs. For a comprehensive study on deterministic spectrum auctions and their applications, the reader is specifically referred to \cite{Maharajan2011, Bogucka2012, Zhang2013, Tragos2013}. 

In this paper, we focus on \emph{nondeterministic} spectrum allocation mechanisms, where the moderator does not have complete knowledge of the state of the PU's spectrum (available/busy for use by CRs). In such a scenario, the moderator relies mainly on the prior distributions and spectrum sensing capabilities of the CRs to gain knowledge about the PU's spectrum usage (as in the case of UNII-II devices). Based on sensing decisions and spectrum valuations provided by the CRs, the moderator allocates the spectrum when deemed available, and charges the CRs for providing service as a moderator. This necessitates us to consider the problem of designing a market-based spectrum management mechanism where the traded commodity in the market (spectrum) is available with uncertainty. In designing such a market-based framework, we encounter the following major challenges. First, CRs, being selfish and rational entities, may reveal falsified sensing decisions as well as incorrect valuations in order to make a higher profit. Second, there are several costs involved that contribute negatively to the players' utilities, such as sensing costs at the CRs, and collision costs if the moderator makes erroneous inferences regarding spectrum availability. In such a scenario, how do we incentivize participation of SUs and maintain feasibility? Third, at what price would the trading of the uncertainly available spectrum take place?

In answer to the aforementioned questions, we propose to use auctions~\cite{Book-Krishna}. Specifically, we design an optimal auction mechanism that maximizes the utility (revenue) of the moderator while maintaining individual-rationality (to incentivize CR participation), and, incentive-compatibility by ensuring truthful revelations of sensing decisions as well as valuations of the CRs. To the best of our knowledge, this is the first work to design an optimal spectrum auction in the presence of spectrum uncertainty, as well as the first work that ensures truthfulness in multiple dimensions. 

\subsection{Recent Work}
In the recent past, there have been few works that address spectrum auctions for CR networks in the presence of spectrum uncertainty. In 2012, to the best of our knowledge, we were the first to propose a \emph{nondeterministic} auction mechanism for spectrum allocation \cite{Nadendla2012} that maximizes the moderator's expected utility. However, in \cite{Nadendla2012}, we investigated a restricted solution space that does not guarantee \emph{individual rationality} at the participating CRs. Furthermore, we only addressed truthful revelation of valuations alone in our model in \cite{Nadendla2012}. 

In \cite{Li2013}, Li \emph{et al.} proposed an online \emph{nondeterministic} auction mechanism that maximizes a social-welfare function over a finite time-horizon, where the valuations are asynchronously transmitted to the moderator. In contrast, in our current work, we design an optimal auction that maximizes the expected utility (revenue) of the moderator, to fill-in the void that remains, given that the prior work in \cite{Li2013} has proposed an auction that maximizes ``social welfare" under spectrum uncertainty. Furthermore, we consider a synchronous (offline) mechanism where the CRs reveal their information at the same time. Note that both these auction mechanisms (social-welfare vs. revenue-maximization, online vs. offline) are two diverse perspectives  in any mechanism design, and therefore, they complement each other. The choice between the two auction mechanisms is scenario-specific, and we refer the reader to \cite{Book-Krishna, Blum2004, Blum2005} for a detailed account on tradeoffs between the two diverse perspectives.

Also in 2013, Tehrani \emph{et al.} proposed a first-price auction mechanism \cite{Tehrani2013} which does not guarantee truthful revelations of either CRs' sensing decisions or the valuations. Therefore, CRs may have an incentive to reveal false valuations to the moderator in order to improve their individual utilities.

\subsection{Contributions}
In this paper, we design a novel \emph{nondeterministic}, truthful and a \emph{sealed-bid} auction mechanism called \emph{Optimal Auction under Uncertain Spectrum Availability}, in short, OAUSA, that maximizes the moderator's utility in the presence of erroneous spectrum sensing decisions at the moderator. We assume a cooperative spectrum sensing framework where the spatially distributed CRs sense the spectrum and send processed observations to the moderator, where these processed observations are fused into a global inference. For more details on the benefits of cooperative spectrum sensing such as, how it alleviates the hidden node problem, the reader may refer to \cite{Yucek2009, Akyildiz2011, Axell2012}. Following are the main contributions of the paper:
\begin{itemize}
    \item We design an optimal auction that maximizes the moderator's utility (revenue) in the presence of uncertainly available PU spectrum. Note that the prior work in \cite{Li2013} has addressed an auction framework that maximizes ``social welfare".
    \item Due to the presence of spectrum uncertainty and the costs of participation and collision, the auction mechanism need not always be feasible. So, we derive conditions under which our proposed auction is feasible. We also investigate the conditions under which there is a need for spectrum sensing.
    \item We propose an auction that ensures truthfulness in the revelation of sensing decisions as well as valuations. To the best of our knowledge, we investigate spectrum auction mechanisms that address multi-dimensional truthfulness for the first time. We also provide analytical bounds on the change in the utility due to our strategy-proof fusion rule.
    \item We present theoretical bounds on the network throughput due to our proposed auction. Furthermore, we also show that our auction provides maximum throughput when the CRs have statistically independent and identically distributed (i.i.d.) valuations.
    \item Simulation results are provided to illustrate the dynamics and effectiveness of the proposed auction mechanism under various scenarios.
\end{itemize}

Note that the proposed auction design is not only a contribution to dynamic spectrum allocation in CR networks, but, advances auction theory in general. In fact, the proposed auction is a generalized Myerson's optimal auction \cite{Myerson1981} applicable to stochastic scenarios where the resource is not available with certainty. The proposed auction reduces to Myerson's optimal auction if the moderator has complete knowledge about the state of the PU spectrum with zero collision and participation costs.


\section{System Model \label{sec: System Model}}

\begin{table}[!t]
	\centering
	\begin{tabular}{ccl}
		\hline\hline
		\\[-1.5ex]
		\textbf{Notation} & \textbf{Range} & \textbf{Description} \\ [0.5ex] 
		\hline
		\\[-1.5ex]
		$t_i$ & $[a_i,z_i]$ & True valuation at the $i^{th}$ CR.
		\\ \\[-1.5ex]
		$p_i(\cdot)$ & $\mathbb{R}^+$ & Probability density function of $t_i$
		\\ \\[-1.5ex]
		$F_i(\cdot)$ & $[0,1]$ & Cumulative density function of $t_i$
		\\ \\[-1.5ex]
		$v_i$ & $[a_i,z_i]$ & Revealed valuation at the $i^{th}$ CR
		\\ \\[-1.5ex]
		$w_i$ & $[a_i,z_i]$ & Virtual valuation at the $i^{th}$ CR
		\\ \\[-1.5ex]
		$\psi_i$ & $[0,1]$ & Probability(Fraction) of spectrum allocated to the $i^{th}$ CR.
		\\ \\[-1.5ex]
		$b_i$ & $\mathbb{R}$ & Payments made by the $i^{th}$ CR to the moderator
		\\ \\[-1.5ex]
		$c_p$ & $\mathbb{R}^+$ & Participation cost incurred at the CRs
		\\ \\[-1.5ex]
		$c_{coll}$ & $\mathbb{R}^+$ & Collision cost imposed by the PU to the moderator
		\\ \\[-1.5ex]
		$P_{f_i}$ & $[0,1]$ & Probability of false alarm at the $i^{th}$ CR
		\\ \\[-1.5ex]
		$P_{d_i}$ & $[0,1]$ & Probability of detection at the $i^{th}$ CR
		\\ \\[-1.5ex]
		$Q_f$ & $[0,1]$ & Probability of false alarm at the moderator
		\\ \\[-1.5ex]
		$Q_d$ & $[0,1]$ & Probability of detection at the moderator
		\\ \\[-1.5ex]
		$q_0$ & $[0,1]$ & Joint probability of deciding $H_0$ and PU's state being $H_0$
		\\ \\[-1.5ex]
		$q_1$ & $[0,1]$ & Joint probability of deciding $H_0$ and PU's state being $H_1$
		\\ \\[-1.5ex]
		$U_i$ & $\mathbb{R}$ & Expected utility at the $i^{th}$ CR in  \emph{OAUSA}
		\\ \\[-1.5ex]
		$U_0$ & $\mathbb{R}$ & Expected utility at the moderator in \emph{OAUSA}
		\\ \\[-1.5ex]
		$\hat{U}_i$ & $\mathbb{R}$ & Expected utility at the $i^{th}$ CR in traditional auctions
		\\ \\[-1.5ex]
		$\hat{U}_0$ & $\mathbb{R}$ & Expected utility at the moderator in traditional auctions
		\\ \\[-1.5ex]
		\hline \hline
	\end{tabular}
	\caption{Glossary of Notation}
	\label{Table: Notation}
\end{table}

Consider a network of $N$ CRs which compete for a given PU's spectrum, as shown in Figure \ref{Fig: model}. We assume the presence of a moderator which is responsible for the task of spectrum allocation in our proposed mechanism. In this paper, we assume that the moderator has no knowledge about the true state of the PU activity in the spectrum. But, we assume that the moderator has knowledge regarding the prior distribution of the PU channel state, and the observation model at the CRs. This information is available at the moderator under the premise that the moderator is aware of the existence and the location of PU. Therefore, the moderator has two roles: \emph{fusion} of sensing decisions and \emph{allocation} of available PU spectrum among the CRs in the network.

\begin{figure}[!t]
	\centering
    \includegraphics[width=3in]{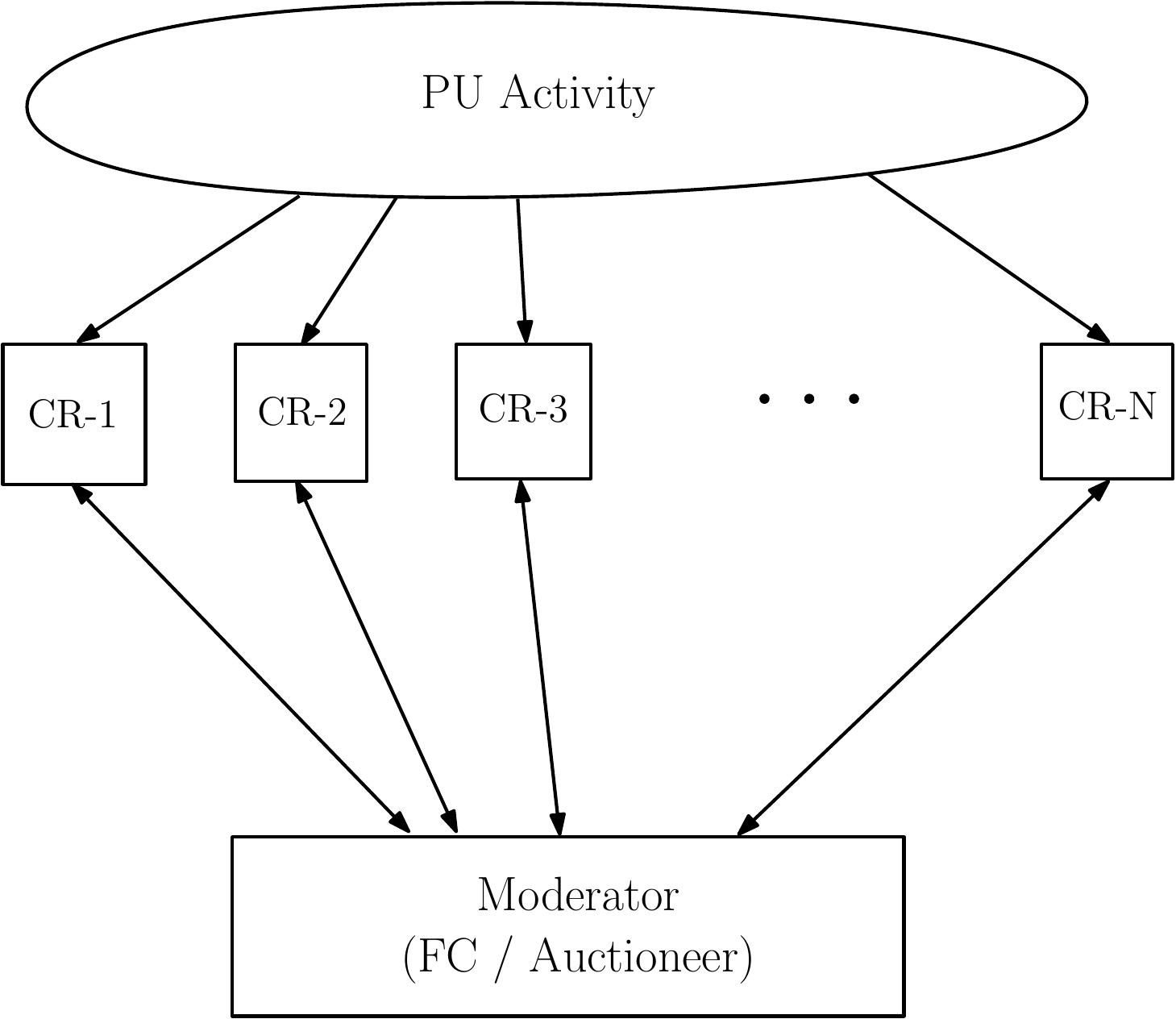}
    \caption{Cognitive-Radio Network Model for collaborative spectrum sensing and allocation}
    \label{Fig: model}
\end{figure}

In order to make reliable decisions about the spectrum allocations with minimal interference to the PU, we assume that the CRs transmit hard-quantized observations (binary decisions) regarding spectrum availability to the moderator, so that the moderator fuses all these sensor messages into a global inference about the PU activity. This role (subsystem) of the moderator is that of a \emph{fusion center} (FC). If the moderator makes an inference that the spectrum is available, then it allocates the spectrum among the set of $N$ CRs in the network via an auction mechanism in which the moderator plays the role of an \emph{auctioneer}.

\emph{In short, the auctioneer subsystem acquires the knowledge of spectrum availability from the FC and allocates the spectrum to the CRs only when it deems it profitable.}

In the remainder of the section, we present the system model in two subsections dedicated to the tasks of sensing-data fusion and the auction of available PU spectrum at the moderator respectively. For the sake of clarity in articulating the role of the moderator, henceforth, we use the terms 'FC' and 'auctioneer' in the place of moderator according to the context and need.

\subsection{Sensing Model}
Spectrum sensing is a binary hypothesis-testing problem where $H_0$ corresponds to the \emph{absence} of the PU activity while $H_1$ corresponds to the \emph{presence} of the PU activity in the licensed band. We assume that the moderator only knows the prior probabilities of the two hypotheses, which are denoted as $\pi_0 = P(H_0)$ and $\pi_1 = P(H_1)$. Note that the secondary network is interested in identifying $H_0$ for spectrum opportunities, and $H_1$ to avoid collisions with the PU.

Let $\mathbb{N} = \{ 1, 2, \cdots, N \}$ denote the set of CRs in the network. We assume that the CRs perform collaborative spectrum sensing by sending their sensing decisions to the FC. The $i^{th}$ CR senses the spectrum and sends a binary decision $u_i \in \{ 0, 1 \}$ to the FC. The FC makes a global decision about the spectrum availability based on all the sensor decisions. In this paper, we assume that the probabilities of false alarm and detection (denoted $P_{f_i}$ and $P_{d_i}$ respectively) of the $i^{th}$ CR are known at the FC, for all $i \in \mathbb{N}$.

For the sake of tractability, we assume that all the CRs in the network reveal their sensing decisions truthfully to the moderator. However, due to the competitive nature of the CRs in acquiring the PU spectrum for their personal needs, they may rationally deviate from revealing their sensing decisions truthfully in order to improve their individual's utilities. Therefore, in Section \ref{sec: Truthful Sensing}, we investigate conditions under which no CR has any incentive to unilaterally deviate from revealing truthful sensing decisions to the moderator. This condition eliminates the possibility of two CRs colliding with each other during their respective secondary transmissions in the PU channel.

Since the FC fuses all the participating CRs' local decisions to obtain a reliable inference about the PU activity, the detection performance at the FC can be characterized by the global probabilities of false alarm and detection, denoted as $Q_f$ and $Q_d$ respectively. Both these quantities can be computed in terms of $P_{f_i}$ and $P_{d_i}$, for all $i \in \mathbb{N}$, for a given fusion rule. Furthermore, we assume that the moderator has knowledge about the observation model at the CRs, which it uses to design the optimal fusion rule. As an example, we present the computation of $Q_f$ and $Q_d$ for a $k$-out-of-$n$ fusion rule in our simulation results in Section \ref{sec: Simulation}, where $k$ is optimal in the sense of error probability.

\subsection{Auction Model}
The moderator may not be willing to provide its service for free (e.g., an IEEE 802.22 BS \cite{IEEE802.22}) and may be a profit seeking entity. Thus, we assume that the CRs have to pay the FC for the service. In any auction, before the process of auctioning starts, the auctioneer valuates the item to be auctioned. Since we assume that the moderator does not have any personal use for spectrum, in this paper, we assume that the moderator's valuation for the available spectrum is zero. Note that the moderator finds it profitable to allocate the available spectrum only when the revenue from the auction is non-negative. Likewise, every CR valuates the available spectrum depending on its need for spectrum. We assume that the $i^{th}$ CR has a truthful valuation (type) $t_i$ for the available spectrum. In this paper, we assume that $t_i$ is a random variable with probability density function $p_i(t_i)$ and cumulative density function $F_i(t_i)$. We assume that each of the valuations $t_i \in [a_i, \ z_i]$ for all $i = 1, \cdots, N$, where $a_i$ and $z_i$ are the minimum and the maximum value that the $i^{th}$ CR can valuate. Of course, when the spectrum is not available for secondary usage (hypothesis $H_1$ is true), we assume that the spectrum is valuated at all the nodes ($N$ CRs, and the moderator) to be zero.

Having valuated the spectrum locally at the CR, the valuations are revealed to the auctioneer (through noiseless, orthogonal control-channels) for spectrum allocation. We also assume that these communication channels are secure, and hence, we restrict ourselves to a sealed-bid framework in our auction design. But, being rational entities, the CRs might reveal false valuations for selfish reasons, in order to acquire greater utility. So, we denote the valuation revealed by the $i^{th}$ CR to the moderator as $v_i \in [a_i, z_i]$ (which may or may not equal to $t_i$). Of course, the choice of $v_i$ depends on the true valuation $t_i$ of the $i^{th}$ CR. Once the spectrum auctioning is complete, the $i^{th}$ CR gets a spectrum allocation $\psi_i \in [0, 1]$ and pays a total amount of $b_i \in \mathbb{R}$ to the moderator at the end, for having participated in the auction. Note that, if the spectrum is indivisible, $\psi_i$ can be interpreted as the probability of allocating the spectrum to the $i^{th}$ CR, while if the spectrum is divisible, $\psi_i$ can be interpreted as the fraction of spectrum allocated to the $i^{th}$ CR.

Note that, in the case where the spectrum-of-interest is treated as an indivisible entity, then valuations can be associated with any model, as the allocation is stochastic in nature. On the other hand, if $\psi_i$ is interpreted as a fractional allocation, then valuations are restricted to linear models. Although valuation models are restricted to a linear structure in the case of divisible spectra, for example, the most relevant and practical example for calculating valuations is based on \emph{throughput}, denoted $\Pi_i$, as shown below.
\begin{equation}
	t_i = \alpha_i \Pi_i = c_i \log \left( 1 + SNR_i \right),
	\label{Eqn: Valuation - Throughput}
\end{equation}
where $\alpha_i$ is a known constant for all $i = 1, \cdots, N$, and, $SNR_i$ is the instantaneous SNR at the corresponding $i^{th}$ CR's receiver. Since the receive SNR is unknown beforehand especially in the context of fading models, $t_i$ is a random variable due to the randomness in $SNR_i$ whose density function can be computed for a given channel model.

Given that the moderator makes all of its decisions based on the CRs' sensing decisions, there is a finite probability with which it can make errors. The error that is of particular concern to the moderator is the case where it allocates the spectrum when the PU is using it. In this case, CRs collide with the PU due to erroneous sensing inference made by the moderator. Therefore, in this paper, we assume that the moderator is responsible for such a mishap regarding which PU penalizes it with a collision cost $c_{coll}$. This is enabled through an underlying protocol where both the PU and the moderator can communicate between themselves and decide the state of the system.

%

Table \ref{Table: Auction} describes the payoffs corresponding to different strategies employed by different players in the auction model, where $c_{coll}$ is the cost of a collision with the PU, and $c_p$ is the cost of participation (includes cost of sensing and transmission) for a given CR. As shown in Table \ref{Table: Auction}, we assume that the CRs are risk-neutral and restrict the structure of the payoffs to be additively separable in terms of payments, spectrum and costs incurred.

\begin{table*}[!t]
	\centering
	\begin{tabular}{c c c c c c}
		\hline\hline
		\textbf{Player} & \textbf{Strategies} & \textbf{Utilities} & \textbf{Moderator's Inference} & \textbf{True Hypothesis} & \textbf{Comments} \\ [0.5ex] 
		\hline \hline
		\\[-1ex]
		&& $\displaystyle \sum_{i = 1}^N b_i$ & $H_0$ & $H_0$ & no collision \\
		&\raisebox{3.5ex}{Allocate (A)} &  $\displaystyle \sum_{i = 1}^N b_i - \sum_{i = 1}^N \psi_i c_{coll}$ & $H_0$ & $H_1$ & collision\\[4ex]
		\raisebox{2ex}{FC} \\[-4ex]
		&& $\displaystyle \sum_{i = 1}^N b_i$ & $H_1$ & $H_0$ & no collision \\
		&\raisebox{3.5ex}{Not Allocate (NA)} & $\displaystyle \sum_{i = 1}^N b_i$ & $H_1$ & $H_1$ & no collision \\[1ex]
		\\
		\hline
		\\		
		&& $\psi_i t_i - b_i - c_p$ & $H_0$ & $H_0$ & no collision \\
		\raisebox{2.5ex}{$i^{th}$ CR} &\raisebox{2.5ex}{Declare $v_i$}\ignore{\raisebox{2.5ex}{Participate (P)} }&  $-b_i -c_p$ & \multicolumn{2}{c}{otherwise}\\[-1ex]
%
%
		\\
		\hline \hline
	\end{tabular}
	\caption{Auction Game Model for Dynamic Spectrum Access in CR Networks}
	\label{Table: Auction}
\end{table*}

Note that, in our model, we impose the penalty of collisions completely on the FC making it responsible for its erroneous decisions, while the CRs bear a participation cost $c_p$ if they participate in the spectrum sensing task, whether or not spectrum allocation is made to them. In addition, the term $\displaystyle \sum_{i = 1}^N \psi_i c_{coll}$ represents the total cost that the FC bears due to the allocation \boldmath$\psi$\unboldmath. Note that, if $\displaystyle \sum_{i = 1}^N \psi_i = 1$, then the moderator incurs the complete cost of collision, $c_{coll}$. 

Figure \ref{Fig: interactions} shows the interactions between the CRs and the moderator in our proposed mechanism. Once the CRs make their observations, they transmit their local decisions $\mathbf{u}$, along with their valuations (revelations) $\mathbf{v}$ to the moderator. With the sensing knowledge from $\mathbf{u}$, the FC subsystem makes a global inference about the PU state and informs the auctioneer subsystem along with the corresponding $Q_f$ and $Q_d$ respectively. After the auctioning process, the moderator provides the allocation vector \boldmath$\psi \ $\unboldmath to the CRs, and the CRs pay $\mathbf{b}$ to the moderator for its service.

\begin{figure}[!t]
	\centering
    \includegraphics[width=3.3in]{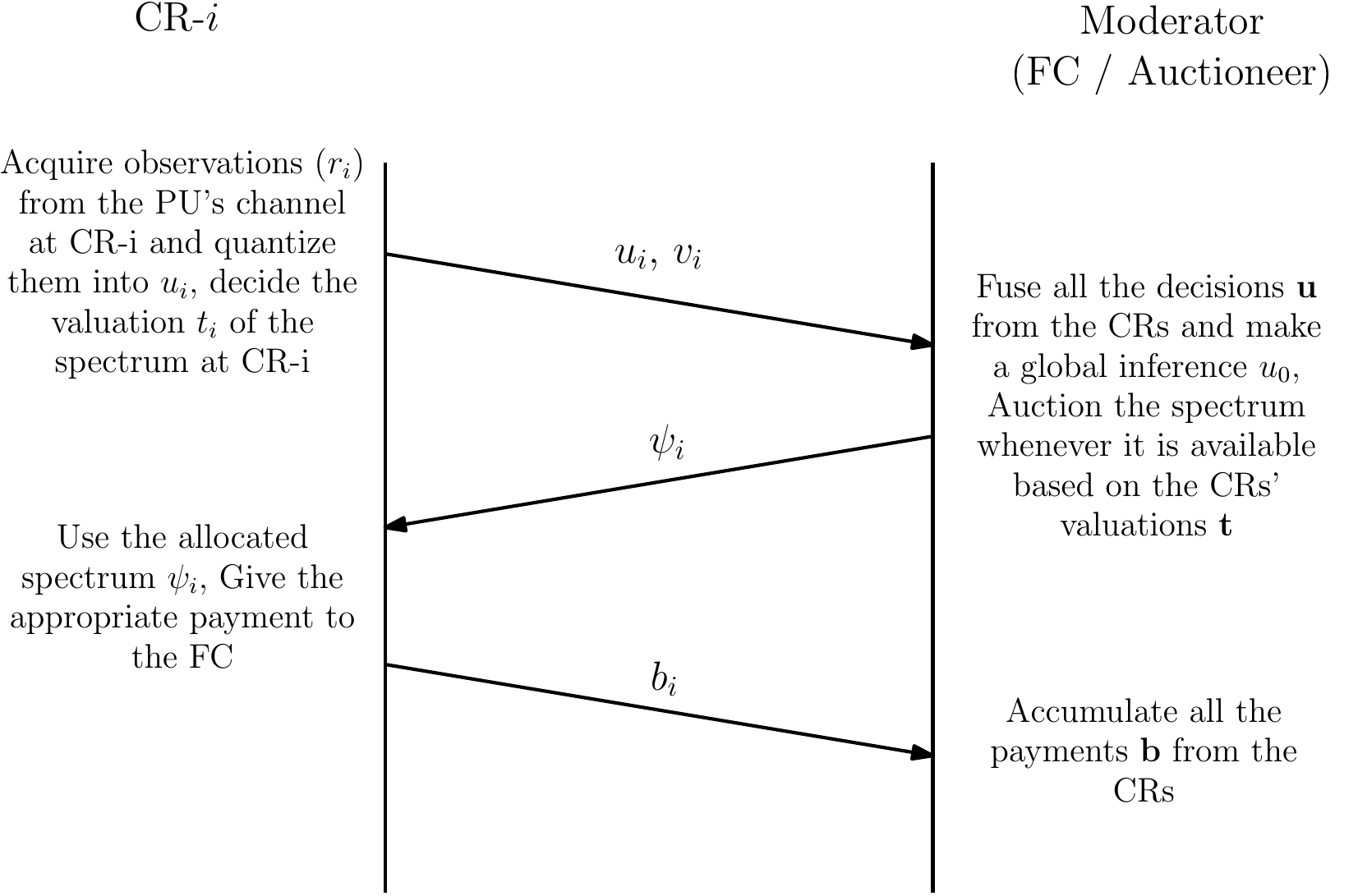}
    \caption{Interactions between the CRs and the FC (Auctioneer)}
    \label{Fig: interactions}
\end{figure}

It is worth mentioning that our auction-based mechanism reduces to several problems of interest as special cases in the context of dynamic spectrum access. Note that if $c_{coll} = 0$, $\pi_0 = 1$, $Q_f = 0$ and $Q_d = 1$, then the auction mechanism reduces to the problem of spectrum commons, where the CRs share unlicensed spectrum that is available with certainty. On the other hand, one can interpret \boldmath$\psi \ $\unboldmath as a fractional or a stochastic allocation vector depending on whether the spectrum is divisible or indivisible respectively. In the case of divisible spectrum, the CRs would employ an OFDM-based system so that an appropriate waveform is designed within the allocated bandwidth. Otherwise, a CR employs a simple spectrum-overlay system when the spectrum is indivisible.


%

\subsection{Some Definitions: Expected Utilities, Rationality, Truthfulness and Feasibility \label{sec: Definitions}}

Being a rational entity, the moderator would want to maximize its net utility. Since the moderator is also responsible for allocating the spectrum, it also incurs the losses whenever erroneous decisions are made regarding the availability of the spectrum, resulting in collisions with the PU. Obviously, if the PU's presence is detected, then there is no motivation for the moderator to allocate the spectrum to the CRs. On the other hand, when PU's absence is detected, the moderator allocates the spectrum to the CRs, if profitable. While collisions with the PU transmissions occur when there is a misdetection at the FC with a probability $q_1 = \pi_1(1 - Q_d)$, allocation without collision takes place with probability $q_0 = \pi_0(1 - Q_f)$. Of course, allocation does not take place with a probability $1 - q_0 - q_1 = \pi_1 Q_d + \pi_0 Q_f$.

From Table \ref{Table: Auction}, we define the expected utility of the moderator as follows.

\begin{defn}
	The utility of the moderator, denoted as $U_0$, is given by Equation (\ref{Eqn: Moderator Expected Utility}),
	\begin{equation}
		\begin{array}{lcl}
		U_0(\mbox{\boldmath$\psi$\unboldmath}, \mathbf{b}) & = & \displaystyle \mathbb{E}_{\mathbf{t}} \left[ \ \pi_0(1 - Q_f) \left\{ \sum_{i=1}^N b_i \right\} \right.
		\\[3ex]
		&& \qquad \displaystyle + \pi_1(1 - Q_d) \left\{ \sum_{i=1}^N b_i - \sum_{i = 1}^N \psi_i c_{coll} \right\}
		\\[3ex]
		&& \qquad \displaystyle \left. + \pi_0 Q_f \left\{ \sum_{i=1}^N b_i \right\} + \pi_1 Q_d \left\{ \sum_{i=1}^N b_i \right\} \ \right]
		\\[4ex]
		& = & \displaystyle  \mathbb{E}_{\mathbf{t}} \left[ \sum_{i=1}^N b_i  - q_1 c_{coll} \sum_{i = 1}^N \psi_i \right].
	\end{array}
	\label{Eqn: Moderator Expected Utility}
	\end{equation}
\end{defn}

Note that the moderator is not aware of the truthful valuations $\mathbf{t}$, since the CRs' revelations indicate that their valuations are $\mathbf{v}$. Therefore, we assume that the truthful valuations $\mathbf{t}$ form a random vector, and therefore, the expected utility of the moderator is computed by averaging over different values of $\mathbf{t}$.

Since the valuations are revealed confidentially to the auctioneer, this mode-of-operation is termed \emph{direct-revelation} \cite{Myerson1981}. We address the notion of direct revelation in greater detail, along with the notion of a given node's rationality in the following subsection.

Note that the moderator's utility depends on $\mathbf{b}$ and \boldmath$\psi$\unboldmath, which in turn, depend on the revelations $\mathbf{v}$, $c_p$, $c_{coll}$, $q_0$ and $q_1$. If the moderator receives a negative expected utility, there is no motivation for the auction mechanism to be pursued by the moderator.


Since a CR does not have any direct-control in the decision-making process, the node will be satisfied as long as its expected utility is non-negative. We define the expected utility of the $i^{th}$ CR, from Table \ref{Table: Auction}, as follows.
\begin{defn}
	The utility of the $i^{th}$ CR, denoted as $U_i$, is given by
	\begin{equation}
		\begin{array}{lcl}
			U_i ( \mbox{\boldmath$\psi$\unboldmath}, \mathbf{b}, t_i) & = & \displaystyle \mathbb{E}_{\mathbf{t_{-i}}}  \left[ q_0 \left\{ \psi_i t_i - b_i - c_p \right\} \right.
			\\[2ex]
			&& \qquad \qquad \displaystyle \left. + (1 - q_0)\left\{ - b_i - c_p \right\} \right]
			\\[3ex]
			& = & \displaystyle \mathbb{E}_{\mathbf{t_{-i}}} \left[ q_0 \psi_i t_i - b_i - c_p \right],
		\end{array}
		\label{Eqn: CR Expected Utility}
	\end{equation}
	where $\mathbf{t}_{-i} = \{ t_1, \cdots, t_{i-1}, t_{i+1}, \cdots, t_N \}$, $\psi_i$ is the fraction of bandwidth allocated to the $i^{th}$ CR, and $b_i$ is the price paid by the $i^{th}$ CR to the FC for participating in the auction.
\end{defn}

Also, we define the expected amount of spectrum that the $i^{th}$ CR gets, in the proposed allocation, is defined as follows.
\begin{defn}
	The expected fraction of bandwidth that node $i$ gets for a given valuation $v_i$, denoted as $\Psi_i(v_i)$ is defined as follows.
	\begin{equation}
		\Psi_i(v_i) = \displaystyle \int q_0 \psi_i(v_i, \mathbf{t_{-i}}) p_{-i}(\mathbf{t_{-i}}) d\mathbf{t_{-i}},
	\end{equation}
	where $p_{-i}(\mathbf{t_{-i}})$ is the joint distribution of $\mathbf{t}_{-i}  = \{ t_1, \cdots, t_{i-1}, t_{i+1}, \cdots, t_N \}$.
	\label{Defn: Expected allocation - vi}
\end{defn}

Since the valuation of the $i^{th}$ CR is known locally, it is not treated as a random variable, and the expectation in the $i^{th}$ CR's utility is carried over $\mathbf{t}_{-i}$. In other words, we assume that the true spectrum valuation (personal preferences) of any given CR is not known to the other CRs.


In order to incentivize participation of CRs in the auction mechanism, there is a need to guarantee a non-negative utility to each of the nodes. Otherwise, the CR would not participate because of the loss that it may incur. This condition is termed \emph{individual rationality}, and is defined as follows.

\begin{defn}[Individual Rationality]
	Individual rationality criterion, which motivates the participation of CRs in the auction mechanism, is defined as
	\begin{equation}
		U_i ( \mbox{\boldmath$\psi$\unboldmath}, \mathbf{b}, t_i) \geq 0, \quad \forall \ i \in \mathbb{N},
		\label{Eqn: Individual Rationality}
	\end{equation}
	for any $t_i \in [a_i,z_i]$.	
\end{defn}

Any selfish and rational entity tries to maximize its individual payoff. So can a CR lie in revealing its valuations and achieve personal gains. However, such a behavior on the part of the CRs can lead to inefficient auction outcomes.  Therefore, in order to prevent the CRs from lying about their valuations, we design the optimal auction mechanism with truthful revelations from the CRs. But, in order to force all the CRs to reveal truthful valuations to the moderator, we need to ensure that there is no incentive for the CR to falsify its revelations. This is called the \emph{incentive compatibility} condition and is defined as follows.


\begin{defn}[Incentive Compatibility]
	A truthful revelation is ensured by the incentive compatibility condition, where no CR has the incentive to reveal false valuations to the moderator. This is given by
	\begin{equation}
		\begin{array}{lcl}
			\displaystyle U_i (\mbox{\boldmath$\psi$\unboldmath}, \mathbf{b}, t_i) & \geq & \displaystyle \int \left[ q_0 \psi_i(v_i, \mathbf{t_{-i}}) t_i - b_i(v_i, \mathbf{t_{-i}})  - c_p \right] \cdot
			\\
			&& \qquad \qquad \qquad \qquad \qquad \quad p_{-i}(\mathbf{t_{-i}}) d\mathbf{t_{-i}},
		\end{array}
		\label{Eqn: Incentive Compatibility}
	\end{equation}
	where, $t_i$ is the true valuation of node i, and $v_i \neq t_i$ is any dishonest valuation declared.
\end{defn}

Note that, if the moderator finds the auction mechanism unprofitable ($U_0$ is smaller than $0$), for all theoretical purposes, we assume that the auction is infeasible. Of course, we consider the possibility of having a feasible auction that allows $\mbox{\boldmath$\psi$\unboldmath} = \mathbf{0}$. 

Thus, we define the feasibility of an auction as follows.

\begin{defn}[Feasibility]
	We define that an auction is \emph{feasible} if the following three conditions are satisfied.
	\begin{itemize}
		\item Individual Rationality, as given by Equation (\ref{Eqn: Individual Rationality})
		\\
		\item Incentive Compatibility, as given by Equation (\ref{Eqn: Incentive Compatibility})
		\\
		\item $\displaystyle \sum_{i \in \mathbb{N}} \psi_i \leq 1$, where $0 \leq \psi_i \leq 1$ for all $i \in \mathbb{N}$.
		\\
		\item $\displaystyle U_0(\mbox{\boldmath$\psi$\unboldmath}, \mathbf{b}) \geq 0$.
	\end{itemize}
	\label{Defn: Feasibility}
\end{defn}

The following lemma, often called the \emph{revelation principle}, discusses the existence of an auction design that ensures the incentive compatibility condition. This result was first presented by Myerson in \cite{Myerson1979}.

\begin{lemma}[Revelation Principle]
	Given any feasible auction mechanism, there exists an equivalent feasible direct revelation mechanism which provides the same utilities to all the players (in this context, CRs and the moderator).
\end{lemma}

This lemma motivates us to proceed further and investigate an optimal direct-revelation mechanism that is feasible.

\section{Optimal Design of Auction Mechanisms Under Spectrum Availability Uncertainty \label{sec: Optimality}}

In this section, we formulate the problem of designing an optimal auction-based mechanism, and design an auction-mechanism that maximizes the expected utility of the moderator $U_0$, while simultaneously holding the incentive-compatibility condition. We formally state our problem of finding the optimal auction mechanism for DSA in CR networks, in Problem \ref{Prob. Statement: Framework-1}. 

\begin{prob-statement}
\begin{align*}
	\begin{array}{l}
		\displaystyle \maximize_{\mbox{\boldmath$\psi$\unboldmath}, \mathbf{b}} U_0(\mbox{\boldmath$\psi$\unboldmath}, \mathbf{b})
		\\[3ex]
		\st
		\\[2ex]
		\text{1. }~ U_i (\mbox{\boldmath$\psi$\unboldmath}, \mathbf{b}, t_i) \geq 0, \quad \forall i \in \mathbb{N}. \mbox{ (As given in Equation (\ref{Eqn: Individual Rationality}))}
		\\[2ex]
		\text{2. }~ U_i (\mbox{\boldmath$\psi$\unboldmath}, \mathbf{b}, t_i) \geq \displaystyle \int \left[ q_0 \psi_i(v_i, \mathbf{t_{-i}}) t_i - b_i(v_i, \mathbf{t_{-i}}) - c_p \right] \cdot
        \\
        \qquad \qquad \qquad \qquad \qquad \qquad \qquad \qquad p_{-i}(\mathbf{t_{-i}}) d\mathbf{t_{-i}},
        \\
        \qquad \quad \forall \  v_i \in [a_i,z_i], \  \forall \  i \in \mathbb{N}. \mbox{ (As given in Equation (\ref{Eqn: Incentive Compatibility}))}
        \\[2ex]
        \text{3. }~ \psi_i \geq 0, \quad \forall i \in \mathbb{N} \ \mbox{and } \displaystyle \sum_{i \in \mathbb{N}} \psi_i \leq 1.
	\end{array}
\end{align*}
\label{Prob. Statement: Framework-1}
\end{prob-statement}

Note that the conditions under which we maximize the utility of the moderator are the same as the ones presented in Definition \ref{Defn: Feasibility} on the feasibility of a given auction mechanism, except for one condition where $U_0 \geq 0$. Therefore, we first present an auction mechanism that satisfies \emph{individual rationality} and \emph{incentive compatibility} conditions. Later, we investigate the feasibility of the proposed auction mechanism in Theorem \ref{Thrm: Feasibility}.

We start our analysis of Problem \ref{Prob. Statement: Framework-1} by investigating the structure of the CRs' utility functions. For the sake of tractability, we restrict our analysis to a certain class of utility functions which are monotonically increasing in terms of their true valuations.

\begin{lemma}
	The expected utility function of the $i^{th}$ player, if monotonically increasing in its true valuation $v_i$, can be expressed as
	\begin{equation}
		U_i (\mbox{\boldmath$\psi$\unboldmath}, \mathbf{b}, t_i) = U_i (\mbox{\boldmath$\psi$\unboldmath}, \mathbf{b}, a_i) + \displaystyle \int_{a_i}^{t_i} \Psi_i(v_i) dv_i.
		\label{Eqn: U_i(v_i) = U_i(a_i) + ...}
	\end{equation}
	\label{Lemma: Properties of U_i}
\end{lemma}

\begin{proof}
	Let us initially consider the case where $a_i \leq v_i \leq t_i \leq b_i$. Expanding the RHS of Equation (\ref{Eqn: Incentive Compatibility}), we have Equation (\ref{Eqn: U_i(v_i) >=  U_i(s_i) + ---}).
	
\begin{equation}
		\begin{array}{lcl}
			U_i (\mbox{\boldmath$\psi$\unboldmath}, \mathbf{b}, t_i) & \geq & \displaystyle \int \left[ q_0 \psi_i(v_i, \mathbf{t_{-i}}) t_i - b_i(v_i, \mathbf{t_{-i}}) - c_p \right] \cdot
			\\[2ex]
			&& \qquad \qquad \qquad \qquad \qquad \qquad \displaystyle p_{-i}(\mathbf{t_{-i}}) d\mathbf{t_{-i}}
			\\[3ex]
			& = & \displaystyle \int \left[ q_0 \psi_i(v_i, \mathbf{t_{-i}}) (v_i + t_i - v_i) \right.
			\\[2ex]
			&& \qquad \qquad \displaystyle \left. - b_i(v_i, \mathbf{t_{-i}}) - c_p \right] p_{-i}(\mathbf{t_{-i}}) d\mathbf{t_{-i}}
			\\[3ex]
			& = & \displaystyle U_i(\mbox{\boldmath$\psi$\unboldmath}, \mathbf{b}, v_i) + (t_i - v_i) \Psi_i(v_i).
		\end{array}
		\label{Eqn: U_i(v_i) >=  U_i(s_i) + ---}
\end{equation}
	
	In the case where $a_i \leq t_i \leq v_i \leq b_i$, Equation (\ref{Eqn: U_i(v_i) >=  U_i(s_i) + ---}) reduces to
	\begin{equation}
		\begin{array}{c}
			U_i (\mbox{\boldmath$\psi$\unboldmath}, \mathbf{b}, v_i) \geq U_i(\mbox{\boldmath$\psi$\unboldmath}, \mathbf{b}, t_i) + (v_i - t_i) \Psi_i(t_i).
		\end{array}
		\label{Eqn: U_i(v_i) >=  U_i(s_i) + --- 2}
	\end{equation}
	
	Combining Equations (\ref{Eqn: U_i(v_i) >=  U_i(s_i) + ---}) and (\ref{Eqn: U_i(v_i) >=  U_i(s_i) + --- 2}), we have	
	\begin{equation}
		\begin{array}{c}
			(t_i - v_i)\Psi_i(v_i) \ \leq U_i (\mbox{\boldmath$\psi$\unboldmath}, \mathbf{b}, t_i) - U_i (\mbox{\boldmath$\psi$\unboldmath}, \mathbf{b}, v_i) 
			\\[1ex]
			\qquad \qquad \qquad \qquad \qquad \leq \ (t_i - v_i)\Psi_i(t_i).
		\end{array}
		\label{Eqn: U_i inequality}
	\end{equation}
	If the expected utility function $U_i$ is a monotonically increasing function of the true valuation $t_i$, then Equation (\ref{Eqn: U_i inequality}) allows $U_i$ to be Riemann-integrable, resulting in Equation (\ref{Eqn: U_i(v_i) = U_i(a_i) + ...}).
\end{proof}

Next, we focus our attention on the properties of $U_0$.
\begin{lemma}
	The expected utility of the FC can be expressed as
	\begin{equation}
			U_0 (\mbox{\boldmath$\psi$\unboldmath}, \mathbf{b}) = \displaystyle \displaystyle - \sum_{i=1}^N U_i(\mbox{\boldmath$\psi$\unboldmath}, \mathbf{b}, a_i) + T,
	\end{equation}
	where $$T = \displaystyle \mathbb{E} \left[ \sum_{i=1}^N q_0 \psi_i \left\{ t_i - \frac{1 - F_i(t_i)}{p_i(t_i)} - \frac{q_1}{q_0} c_{coll} \right\} - N c_p \right],$$ and $F_i(t_i) = \displaystyle \int_{t_i}^{b_i} p_i(s_i) ds_i$.
	
	\label{Lemma: Properties of U_fc}
\end{lemma}

\begin{proof}
	Let us start from the definition of $U_0$, as given in Equation \eqref{Eqn: Moderator Expected Utility}. Therefore, we have
	\begin{equation}
			U_0 (\mbox{\boldmath$\psi$\unboldmath}, \mathbf{b}) = \mathbb{E}_{\mathbf{t}} \left[ \sum_{i=1}^N b_i  - q_1 c_{coll} \sum_{i = 1}^N \psi_i \right].
			\label{Eqn: U_0 defn}
	\end{equation}
	
	Adding and subtracting $\displaystyle \sum_{i=1}^N \left( q_0 \psi_i t_i - c_p \right)$ in the RHS of Equation \eqref{Eqn: U_0 defn}, we have	
\begin{equation}
		\begin{array}{lcl}
			U_0 (\mbox{\boldmath$\psi$\unboldmath}, \mathbf{b}) & = & \displaystyle \mathbb{E}_{\mathbf{t}} \left[ \sum_{i=1}^N \left( c_p + b_i - q_0 \psi_i t_i \right) \right.
			\\[2ex]
			&& \qquad \left. \displaystyle - q_1 c_{coll} \sum_{i = 1}^N \psi_i - N c_p + q_0 \sum_{i=1}^N \psi_i t_i \right]
			\\[3ex]
			& = & \displaystyle - \sum_{i=1}^N \int U_i(\mbox{\boldmath$\psi$\unboldmath}, \mathbf{b}, t_i) p_i(t_i) dt_i
			\\[2ex]
			&& \qquad \left. \displaystyle + \mathbb{E} \left[  \sum_{i=1}^N \psi_i \left\{q_0 t_i -  q_1 c_{coll} \right\} - N c_p \right] \right. .
		\end{array}
		\label{Eqn: Expanded U_0}
	\end{equation}
	
	Substituting the result of Lemma \ref{Lemma: Properties of U_i} in Equation (\ref{Eqn: Expanded U_0}), we have
	\begin{equation}
		\begin{array}{lcl}
			U_0(\mbox{\boldmath$\psi$\unboldmath}, \mathbf{b}) & = & \displaystyle - \sum_{i=1}^N U_i(\mbox{\boldmath$\psi$\unboldmath}, \mathbf{b}, a_i) + T,
		\end{array}
		\label{Eqn: Expanded U_{fc} 2}
	\end{equation}
	where
	\begin{equation}
		\begin{array}{lcl}
			T & = & \displaystyle \mathbb{E}_{\mathbf{t}} \left[ q_0 \sum_{i=1}^N \psi_i t_i - q_1 c_{coll} \sum_{i = 1}^N \psi_i - N c_p \right]
			\\[2ex]
			&& \qquad \displaystyle - \sum_{i=1}^N \int_{a_i}^{b_i} \left( \int_{a_i}^{t_i} \Psi(s_i) ds_i \right) p_i(t_i) dt_i.
		\end{array}
		\label{Eqn: T}
	\end{equation}
	
	Henceforth, we focus our attention to the structure of $T$, as defined in Equation \eqref{Eqn: T}. Changing the order of integration in the second term of the R.H.S of Equation (\ref{Eqn: T}), and substituting $\pi_1$ in place of $1 - \pi_0$, we have Equation (\ref{Eqn: T-2}).
\begin{figure*}[!t]
\normalsize
\setcounter{MYtempeqncnt}{\value{equation}}
	\begin{equation}
		\begin{array}{lcl}
			T & = & \displaystyle \mathbb{E} \left[ q_0 \sum_{i=1}^N \psi_i t_i - q_1 c_{coll} \sum_{i = 1}^N \psi_i - N c_p \right] - \sum_{i=1}^N \int_{a_i}^{b_i} \left( \int_{s_i}^{b_i} \Psi(s_i) p_i(t_i) dt_i \right) ds_i.
			\\
			\\
			& = & \displaystyle \mathbb{E} \left[ q_0 \sum_{i=1}^N \psi_i t_i - q_1 c_{coll} \sum_{i = 1}^N \psi_i - N c_p \right] - \sum_{i=1}^N \int_{a_i}^{b_i} (1 - F_i(s_i)) \Psi(s_i) ds_i
			\\
			\\
			& = & \displaystyle \mathbb{E} \left[ q_0 \sum_{i=1}^N \psi_i t_i -  q_1 c_{coll} \sum_{i = 1}^N \psi_i - N c_p \right] - \sum_{i=1}^N \int q_0 \psi_i \left(\frac{1 - F_i(t_i)}{p_i(t_i)}\right) p(\mathbf{t}) d\mathbf{t}.
		\end{array}
		\label{Eqn: T-2}
	\end{equation}
\addtocounter{MYtempeqncnt}{1}
\setcounter{equation}{\value{MYtempeqncnt}}
\hrulefill
\vspace*{2pt}
\end{figure*}
	
	Rearranging the terms in Equation (\ref{Eqn: T-2}), we have
	\begin{equation}
		T = \displaystyle \mathbb{E} \left[ \sum_{i=1}^N q_0 \psi_i \left\{ t_i - \frac{1 - F_i(t_i)}{p_i(t_i)} - \frac{q_1}{q_0} c_{coll} \right\} - N c_p \right].
		\label{Eqn: T-3}
	\end{equation}
\end{proof}

Having used the incentive compatibility condition, given in Equation (\ref{Eqn: Incentive Compatibility}) in Lemma \ref{Lemma: Properties of U_fc}, we can rewrite Problem \ref{Prob. Statement: Framework-1} as follows.

\begin{prob-statement}
	\begin{flalign*}
    	\displaystyle \argmax_{\mbox{\boldmath$\psi$\unboldmath}, \mathbf{b}} \quad & \displaystyle T - \sum_{i=1}^N U_i(\mbox{\boldmath$\psi$\unboldmath}, \mathbf{b}, a_i) \quad \mbox{ s.t.}
        \\ & \text{1. }~ U_i (\mbox{\boldmath$\psi$\unboldmath}, \mathbf{b}, t_i) \geq 0, \quad \forall i \in \mathbb{N}. 
        \\ & \text{2. }~ \psi_i \geq 0, \quad \forall i \in \mathbb{N}
        \\ & \text{3. }~ \displaystyle \sum_{i \in \mathbb{N}} \psi_i \leq 1.
    \end{flalign*}
    \label{Prob. Statement: Framework-1-reduced}
\end{prob-statement}

Note that the term $Nc_p$ can be interpreted as the compensation that the FC pays back to the CRs to incentivize their participation in the spectrum sensing task.

Now, we focus our attention on solving Problem \ref{Prob. Statement: Framework-1-reduced}. Using Lemmas \ref{Lemma: Properties of U_i} and \ref{Lemma: Properties of U_fc}, we prove the most important result of this section in the following theorem.

\begin{thrm}
	For all $i \in \mathbb{N}$, if the function $w_i(t_i) = \displaystyle t_i - \frac{1 - F_i(t_i)}{p_i(t_i)}$ is strictly increasing in $t_i$ (\textbf{regularity} condition), then the optimal allocation, that maximizes the moderator's revenue, is given by
	\begin{equation}
		\psi_i^* = \left\{
		\begin{array}{ccl}
			\displaystyle \Delta_i & ; & \mbox{if} \ \ |\mathbb{M}(\mathbf{t})| > 0, \ \forall \ i \in \mathbb{M}(\mathbf{t}).
			\\
			\\
			0 & ; & otherwise
		\end{array} \right.
		\label{Eqn: Optimal Allocation}
	\end{equation}
	for any $\Delta_i$ such that $\displaystyle \sum_{i \in \mathbb{M}(\mathbf{t})} \Delta_i = 1$, and where $\mathbb{M}(\mathbf{t}) = \left\{\displaystyle i \ \left| \ i = \displaystyle \argmax_{j \in \mathbb{N}} w_j(t_j) \geq \frac{q_1}{q_0} c_{coll} \right. \right\}$.
	
	Also, the optimal payments made by the CRs to the FC are given by
	\begin{equation}
		b_i^* = q_0 \psi_i^* t_i - c_p - \displaystyle q_0 \int_{a_i}^{t_i} \psi_i^* \left( s_i, \mathbf{t_{-i}} \right) ds_i.
		\label{Eqn: Optimal Price / Bid}
	\end{equation}
	\label{Thrm: Optimal Auction}
\end{thrm}

\begin{proof}
	Note that, the term $U_i(a_i)$ has a negative contribution to the utility of the moderator. As we know that $U_i(\mbox{\boldmath$\psi$\unboldmath}, \mathbf{b}, t_i)$ is non-negative for all $t_i \in [a_i, z_i]$ and $i \in \mathbb{N}$ from the individual rationality criterion, as given in Equation \eqref{Eqn: Individual Rationality}, the moderator will let $U_i(\mbox{\boldmath$\psi$\unboldmath}, \mathbf{b}, a_i) = 0, \forall i \in \mathbb{N}$. Substituting this in Equation \eqref{Eqn: U_i(v_i) = U_i(a_i) + ...}, we have
	\begin{equation}
		U_i(t_i) = \displaystyle \int_{a_i}^{t_i} \left( \int q_0 \psi_i(v_i, \mathbf{t_{-i}}) p_{-i}(\mathbf{t_{-i}}) d\mathbf{t_{-i}} \right) dv_i.
		\label{Eqn: Optimal payment calc}
	\end{equation}
	
	Substituting the definition of $U_i(t_i)$ (refer to Equation \eqref{Eqn: CR Expected Utility}) in Equation \eqref{Eqn: Optimal payment calc}, we have
	\begin{equation}
		\displaystyle \mathbb{E}_{\mathbf{t}_{-i}} \left[ q_0 \psi_i t_i - b_i - c_p \right] = \displaystyle \mathbb{E}_{\mathbf{t}_{-i}} \left[ q_0 \int_{a_i}^{t_i} \psi_i(v_i, \mathbf{t_{-i}}) dv_i \right].
		\label{Eqn: Optimal payment calc 2}
	\end{equation}
	
	Thus, one of the possible selection of optimal payments is given by
	\begin{equation}
		\displaystyle q_0 \psi_i^* t_i - b_i - c_p  = q_0 \int_{a_i}^{t_i} \psi_i^*(v_i, \mathbf{t_{-i}}) dv_i.
		\label{Eqn: Optimal payment calc 3}
	\end{equation}
	
	Rearranging the terms in Equation \eqref{Eqn: Optimal payment calc 3}, we have Equation \eqref{Eqn: Optimal Price / Bid}.
	
	Driving the term $U_i(a_i) = 0$, the utility of the moderator can be rewritten as follows.
	\begin{equation}
		U_0(\mbox{\boldmath$\psi$\unboldmath}, \mathbf{b}) = T.
		\label{Eqn: Moderator-Utility-3}
	\end{equation}
	
	Therefore, we focus our attention on $T$, which is restated as follows.
	\begin{equation}
		T = \displaystyle \mathbb{E} \left[ \sum_{i=1}^N q_0 \psi_i \left\{ t_i - \frac{1 - F_i(t_i)}{p_i(t_i)} - \frac{q_1}{q_0} c_{coll} \right\} - N c_p \right].
		\label{Eqn: T-3b}
	\end{equation}
	
	Let us denote $w_i(t_i) = \displaystyle t_i - \frac{1 - F_i(t_i)}{p_i(t_i)}$. Note that, if $w_i(t_i)$ is an increasing function of $t_i$ (\emph{regularity} condition), then the moderator maximizes $T$ by allocating the spectrum to the
	CRs in the set $\mathbb{M}(\mathbf{t}) = \left\{\displaystyle i \ \left| \ i = \displaystyle \argmax_{j \in \mathbb{N}} w_j(t_j) \geq \frac{q_1}{q_0} c_{coll} \right. \right\}$. Whenever $|\mathbb{M}(\mathbf{t})| \geq 2$, the tie-breaker can be resolved in any manner without affecting the optimal value of $T$, as long as the allocation of the spectrum is restricted to the nodes in $\mathbb{M}(\mathbf{t})$.	
	

\end{proof}

Since the allocation of the spectrum depends on $w_i(t_i)$, instead of the valuation $t_i$ at the $i^{th}$ CR, we call the term $w_i(t_i)$ as the \emph{virtual} valuation of the $i^{th}$ CR.

Note that the $i^{th}$ CR's payment, as given in Equation (\ref{Eqn: Optimal Price / Bid}), has an integral term $\displaystyle \int_{a_i}^{t_i} \psi_i^* \left( s_i, \mathbf{t_{-i}} \right) ds_i$, which needs to be computed in order to find the payments of the $i^{th}$ CR. The computation of this term is very interesting since it requires us to investigate four cases, as follows. Within our analysis, we denote $t_*$ as the valuation that corresponds to the second highest element (in value) within the array of virtual valuations $\mathbf{w}(\mathbf{t})$.

\paragraph{CASE-1}[$i \notin \mathbb{M}(\mathbf{t})$]
From Equation (\ref{Eqn: Optimal Allocation}) in Theorem \ref{Thrm: Optimal Auction}, we have $\psi_i^*(t_i, \mathbf{t_{-i}}) = 0$. Therefore, $\displaystyle \int_{a_i}^{t_i} \psi_i^* \left( s_i, \mathbf{t_{-i}} \right) ds_i = 0$, and consequently,
\begin{equation}
	b_i = - c_p.
	\label{Eqn: Optimal Payment - CASE1}
\end{equation}

\paragraph{CASE-2}[$i \in \mathbb{M}(\mathbf{t})$, $|\mathbb{M}(\mathbf{t})| \geq 2$]
As shown in Equation (\ref{Eqn: Optimal Allocation}), we have $\psi_i^*(t_i, \mathbf{t_{-i}}) = \Delta_i$. But, if the $i^{th}$ CR chooses to deviate to a valuation smaller than $t_i$, then $i \ni \mathbb{M}(\mathbf{t})$ and $\displaystyle \int_{a_i}^{t_i} \psi_i^* \left( s_i, \mathbf{t_{-i}} \right) ds_i = 0$. Therefore,
\begin{equation}
	b_i = q_0 \Delta_i t_i - c_p.
	\label{Eqn: Optimal Payment - CASE2}
\end{equation}

\paragraph{CASE-3}[$i \in \mathbb{M}(\mathbf{t})$, $|\mathbb{M}(\mathbf{t})| = 1$, $t_* \geq a_i$]
If $t_* \geq a_i$, then the $i^{th}$ CR gets the spectrum as long as $t_i \geq t_*$. In other words,
$$\displaystyle \int_{a_i}^{t_i} \psi_i^* \left( s_i, \mathbf{t_{-i}} \right) ds_i = \int_{t_*}^{t_i} 1 \cdot ds_i = (t_i - t_*).$$ Therefore, the payment at the $i^{th}$ CR can be calculated as
\begin{equation}
	\begin{array}{lcl}
		b_i & = & q_0 t_i - c_p - q_0 (t_i - t_*)
		\\
		& = & q_0 t_* - c_p.
	\end{array}
	\label{Eqn: Optimal Payment - CASE3}
\end{equation}

\paragraph{CASE-4}[$i \in \mathbb{M}(\mathbf{t})$, $|\mathbb{M}(\mathbf{t})| = 1$, $t_* < a_i$]
Since $t_* < a_i$, then the $i^{th}$ CR gets the spectrum for any $t_i \in [a_i, b_i]$. In other words,
$$\displaystyle \int_{a_i}^{t_i} \psi_i^* \left( s_i, \mathbf{t_{-i}} \right) ds_i = (t_i - a_i).$$ Therefore, the payment at the $i^{th}$ CR can be calculated as
\begin{equation}
	\begin{array}{lcl}
		b_i & = & q_0 t_i - c_p - q_0 (t_i - a_i)
		\\
		& = & q_0 a_i - c_p.
	\end{array}
	\label{Eqn: Optimal Payment - CASE4}
\end{equation}

Combining Equations (\ref{Eqn: Optimal Payment - CASE1})-(\ref{Eqn: Optimal Payment - CASE4}), the $i^{th}$ CR's payment to the moderator is summarized as follows.
\begin{equation}
	b_i =
	\begin{cases}
		\ - c_p; & \mbox{if } i \notin \mathbb{M}(\mathbf{t})
		\\
		\ q_0 \Delta_i t_i - c_p & \mbox{if } i \in \mathbb{M}(\mathbf{t}) \mbox{ and } |\mathbb{M}(\mathbf{t})| \geq 2
		\\
		\ q_0 t_* - c_p & \mbox{if } i \in \mathbb{M}(\mathbf{t}), \ |\mathbb{M}(\mathbf{t})| = 1 \mbox{ and } t_* \geq a_i
		\\
		\ q_0 a_i - c_p & \mbox{if } i \in \mathbb{M}(\mathbf{t}), \ |\mathbb{M}(\mathbf{t})| = 1 \mbox{ and } t_* < a_i.
	\end{cases}
	\label{Eqn: Optimal Payment - Final}
\end{equation}

Having computed the CRs' payments, in the case where $|\mathbb{M}(\mathbf{t})| \geq 2$, if the allocation is chosen to be $\Delta_i = \displaystyle \frac{1}{|\mathbb{M}(\mathbf{t})|}$, then, the proposed auction mechanism can be summarized as an algorithm, denoted \emph{Optimal Auction under Uncertain Spectrum Availability} (in short, OAUSA), in Figure \ref{Algo: SpecAlloc}. In case, the designer would like to choose a different allocation \boldmath$\Delta$\unboldmath, appropriate changes can be made by replacing Line 16 in Figure \ref{Algo: SpecAlloc}.

\begin{figure}[t]
	\centering
	\begin{algorithmic}[1]
		\Procedure{OAUSA}{$\mathbf{t}$, $\mathbf{u}$}
		\State Fuse $\mathbf{u}$ into a global inference $u_0$
		\If{$u_0 = 1$}
			\ForAll{$i \in \mathbb{N}$}
				\State $\psi_i \gets 0$
				\State $b_i \gets -c_p$
			\EndFor
		\Else
			\State $w_i(t_i) \gets \displaystyle t_i - \frac{1 - F_i(t_i)}{p_i(t_i)}$, for all $i = 1, \cdots, N$.
			\State Find $\mathbb{M}(\mathbf{t}) = \{ i | i = \argmax_{j \in \mathbb{N}} w_j(t_j) \}$.
			\ForAll{$i \in \mathbb{N}$}
				\If{$i \notin \mathbb{M}(\mathbf{t})$}
					\State $\psi_i \gets 0$
					\State $b_i \gets - c_p$
				\ElsIf{$i \in \mathbb{M}(\mathbf{t})$, $|\mathbb{M}(\mathbf{t})| \geq 2$}
					\State $\Delta_i \gets \displaystyle \frac{1}{|\mathbb{M}(\mathbf{t})|}$
					\State $\psi_i \gets \Delta_i$
					\State $b_i \gets q_0 \Delta_i t_i - c_p$
				\ElsIf{$i \in \mathbb{M}(\mathbf{t})$, $|\mathbb{M}(\mathbf{t})| = 1$}
					\State $\psi_i \gets 1$.
					\State Find $t_*$ such that $w(t_*)$ is the second largest in the vector $\mathbf{w}(\mathbf{t})$.
					\If{$t_* \geq a_i$}
						\State $b_i \gets q_0 t_* - c_p$.
					\Else
						\State $b_i \gets q_0 a_i - c_p$.
					\EndIf
				\EndIf
			\EndFor
		\EndIf		
		\State \textbf{return} (\boldmath$\psi$\unboldmath, $\mathbf{b}$) 
		\EndProcedure
	\end{algorithmic}
	\caption{Pseudo-code for the proposed algorithm to find the optimal spectrum allocation at the moderator}
	\label{Algo: SpecAlloc}
\end{figure}

\section{Feasibility and Throughput Analysis \label{sec: Discussion}}

\subsection{Feasibility of the Proposed Auction \label{sec: Feasibility}}
	Note that the proposed auction in Theorem \ref{Thrm: Optimal Auction} is feasible only if $U_0 \geq 0$. Otherwise, it is more profitable for the moderator to have the spectrum for itself. In other words, from Equation \eqref{Eqn: Moderator-Utility-3}, we need $T \geq 0$. Therefore, if
	$$\displaystyle w_{max} =
	\begin{cases}
		w_i(t_i); & \mbox{ if } i = \displaystyle \argmax_{j \in \mathbb{N}} w_j(t_j)
		\\
		0; & \mbox{otherwise,}
	\end{cases}$$
	then it is expected that
	\begin{equation}
		\begin{array}{lcl}
			U_0 & = & T 
			\\[1ex]
			& = & q_0 \left[ \mathbb{E}[w_{max}] - \frac{q_1}{q_0} c_{coll} \right] - N c_p \geq 0.
		\end{array}
		\label{Eqn: U0 - final}
	\end{equation}
	
	Or, equivalently,
	\begin{equation}
		\displaystyle \mathbb{E}[w_{max}] \geq \displaystyle \frac{1}{q_0} N c_p + \frac{q_1}{q_0} c_{coll}.
		\label{Eqn: Favorable Condition}
	\end{equation}
	
	In other words, the proposed auction is feasible only when the expected value of the maximum valuation among all the CRs is greater than $\displaystyle \frac{1}{q_0} N c_p + \frac{q_1}{q_0} c_{coll}$. Note that the instantaneous value of $w_{max}$ need not necessarily be greater than $\displaystyle \frac{1}{q_0} N c_p + \frac{q_1}{q_0} c_{coll}$. Therefore, we have the following theorem.
	
\begin{thrm}
	The auction proposed in Theorem \ref{Thrm: Optimal Auction} is feasible only if
	$$\displaystyle \mathbb{E}[w_{max}] \geq \frac{1}{q_0} N c_p + \frac{q_1}{q_0} c_{coll}.$$
	\label{Thrm: Feasibility}
\end{thrm}

Note that, if the $i^{th}$ CR does not get any spectrum, then the optimal choice of the payment is $b_i = -c_p$ which is negative. This simply means that the FC pays an amount $c_p$ back to the $i^{th}$ CR node, as a compensation for sensing. Consequently, the $i^{th}$ CR participates in the spectrum sensing task with no loss or gain locally. But, it improves the sensing performance at the FC and reduces the probability of collision of the moderator with the PU. This allows the moderator to accumulate a non-negative utility, on an average.

\begin{bluetext}
\subsection{Need for Spectrum Sensing \label{sec: Favorable-Conditions}}
Our proposed auction mechanism, OAUSA, is designed under the assumption that the moderator does not have complete knowledge about the availability of PU spectrum. Therefore, collaborative spectrum sensing is incorporated into the design of spectrum auctions, which introduced several new costs into our model. In this subsection, we investigate the merit of cooperative spectrum sensing by comparing our proposed auction mechanism, OAUSA, with traditional auction mechanisms which are not designed for scenarios where the spectrum is available with uncertainty. In traditional spectrum auctions, the moderator believes that the spectrum is always available for auctioning even though it does not have knowledge about uncertain spectrum availability\footnote{This consideration follows from the fact that traditional auction mechanisms assume that spectrum is always available for auctioning.}. Since the CRs do not pursue the task of spectrum sensing, $c_p = 0$, thereby resulting in a collision with the PUs with probability $\pi_1$, since the moderator always allocates the spectrum to the CRs irrespective of whether PU is using the spectrum or not. Let $\hat{U}_0$ and $\hat{U}_i$ denote the moderator's and the $i^{th}$ CR's utilities in the case of a traditional auction mechanism respectively. Then,

\begin{equation}
	\begin{array}{lcl}
		\hat{U}_0(\mbox{\boldmath$\psi$\unboldmath}, \mathbf{b}) & = & \displaystyle \mathbb{E}_{\mathbf{t}} \left[ \sum_{i = 1}^N b_i - \pi_1 c_{coll} \sum_{i = 1}^N \psi_i \right],
		\\
		\hat{U}_i(\mbox{\boldmath$\psi$\unboldmath}, \mathbf{b}, t_i) & = & \displaystyle \mathbb{E}_{\mathbf{t}_{-i}} \left[ \psi_i t_i - b_i \right].
	\end{array}
	\label{Eqn: Utility - Traditional Auctions}
\end{equation}

Having defined the utility of a traditional auction in the presence of uncertainly available spectrum, we investigate the conditions under which collaborative spectrum sensing provides a greater utility in the following lemma.

\begin{lemma}
Collaborative spectrum sensing in OAUSA is profitable only when
\begin{equation}
	q_1 \leq \pi_1 - N \frac{c_p}{c_{coll}}
\end{equation}
\end{lemma}
\begin{proof}
Let us first express the utility of the traditional auction mechanism, as given in Equation (\ref{Eqn: Utility - Traditional Auctions}), in terms of the corresponding utilities in the proposed OAUSA framework, as follows.
\begin{equation}
	\begin{array}{lcl}
		\hat{U}_0(\mbox{\boldmath$\psi$\unboldmath}, \mathbf{b}) & = & \displaystyle U_0(\mbox{\boldmath$\psi$\unboldmath}, \mathbf{b}) - (\pi_1 - q_1) c_{coll} + N c_p,
		\\
		\hat{U}_i(\mbox{\boldmath$\psi$\unboldmath}, \mathbf{b}, t_i) & = & \displaystyle U_i(\mbox{\boldmath$\psi$\unboldmath}, \mathbf{b}, t_i) + (1 - q_0) t_i \mathbb{E}_{\mathbf{t}_{-i}} \left[ \psi_i \right].
	\end{array}
	\label{Eqn: Utility - Traditional Auctions}
\end{equation}

Therefore, the designer has an incentive to adopt OAUSA if $U_0 \geq \hat{U}_0$. In other words, OAUSA is favored only when $(\pi_1 - q_1) c_{coll} - N c_p \geq 0$. 
\end{proof}

Consequently, OAUSA gives a greater utility at the moderator than any traditional auction in scenarios where the cost of collision at the moderator is significantly larger than the cost of participation at the CRs $\left( c_{coll} \geq \frac{N}{\pi_1 - q_1} c_p \right)$. This is true in most practical scenarios since the penalties for colliding with the PU are typically modelled to be very severe, when compared to the instantaneous benefits at the CRs from acquiring a spectrum. On the other hand, the $i^{th}$ CR's utility in the case of traditional auctions is always greater than the proposed OAUSA auction. This deviation is bounded by $(1 - q_0) t_i$, since $\mathbb{E}_{\mathbf{t}_{-i}} \left[ \psi_i \right] \leq 1$.
\end{bluetext}


\begin{magentatext}
\subsection{Throughput Analysis \label{sec: Throughput}}
In this subsection, we investigate the performance of our proposed auction mechanism, OAUSA, from a throughput perspective. We assume that the CRs' valuations are modeled in terms of their respective throughput, as defined in Equation \eqref{Eqn: Valuation - Throughput} when the moderator allocates the spectrum. We denote the throughput at the $i^{th}$ CR as $x_i = \log (1 + SNR_i)$. Therefore, Equation \eqref{Eqn: Valuation - Throughput} can be rewritten as $t_i = c_i x_i$. Furthermore, we assume that $c_i = c$ with $c>0$, for all $i = 1, \cdots, N$. 

Without any loss of generality, let $j \triangleq \argmax \{ w_i \}$ denote the index of the CR that is selected for spectrum allocation by our proposed OAUSA algorithm. Similarly, let $k \triangleq \argmax \{ x_i \}$ denote the index of the CR that gives maximum throughput to the network, when selected. Since $c_i = c$ for all $i = 1, \cdots, N$, we also have $k = \argmax \{ t_i = c x_i \}$, and therefore, VCG auction provides maximum network throughput. Furthermore, we denote the network throughput due to OAUSA and VCG algorithms as $x_{OAUSA} = x_j$ and $x_{VCG} = x_k$ respectively. In the following lemma, we state that OAUSA achieves the same network throughput as VCG, when the CR valuations are statistically independent and identically distributed. Consequently, OAUSA also achieves the maximum network throughput. 

\begin{lemma}
If all the CRs have independent and identically distributed valuations, then $x_{OAUSA} = x_{VCG} \triangleq \displaystyle \max \mathbf{x}$ if the regularity and the feasibility conditions, stated in Theorems \ref{Thrm: Optimal Auction} and \ref{Thrm: Feasibility} respectively, hold true.
\end{lemma}
\begin{proof}
Consider two valuations $t_1 = c x_1$ and $t_2 = c x_2$ sampled independently from the same distribution. Let $t_1 \leq t_2$. Since the regularity function holds, let $w(\cdot)$ be the virtual valuation function. We denote the two virtual valuations as $w_1 = w(t_1)$ and $w_2 = w(t_2)$. Since $w(t)$ is a non-decreasing function of $t$, $w_1 \leq w_2$. As a result, given a set of valuations $\mathbf{t}$ and their corresponding virtual valuations $\mathbf{w}$, $\argmax \mathbf{w} = \argmax \mathbf{t} = \argmax \mathbf{x}$.
\end{proof}

In the case where the CR valuations are not identically distributed, we compare the average throughput attained by OAUSA to that of VCG auction in the following theorem.
\begin{thrm}
Let $j \triangleq \argmax \{ w_i \}$ and $k \triangleq \argmax \{ t_i = c x_i \}$ denote the indices of CRs selected for allocation in OAUSA and VCG auction mechanisms respectively. Then, if the regularity and the feasibility conditions, stated in Theorems \ref{Thrm: Optimal Auction} and \ref{Thrm: Feasibility} respectively, hold true, we have
\begin{equation}
	\displaystyle x_k - \frac{1}{c} \left[ \frac{1 - F(x_k)}{p(x_k)} \right] \leq x_j \leq x_k.
	\label{Eqn: Throughput-Bound-Final}
\end{equation}
whenever $c_i = c$ for all $i = 1, \cdots, N$.
\end{thrm}
\begin{proof}
Given that $j \triangleq \argmax \{ w_i \}$ and $k \triangleq \argmax \{ t_i = c x_i \}$, we have
\begin{equation}
\begin{array}{lcr}
	w_k \leq w_j & \quad \mbox{and} \quad & x_j \leq x_k
\end{array}
\label{Eqn: Inequality-1}
\end{equation}

But, we have
\begin{equation}
	x_i \geq \displaystyle x_i - \frac{1}{c} \left[ \frac{1 - F_i(x_i)}{p_i(x_i)} \right] = \frac{1}{c} w_i
	\label{Eqn: Inequality-2}
\end{equation}

Combining Equations \eqref{Eqn: Inequality-1} and \eqref{Eqn: Inequality-2}, we have
\begin{equation}
	\displaystyle x_k - \frac{1}{c} \left[ \frac{1 - F_k(x_k)}{p_k(x_k)} \right] \ = \ \frac{1}{c} w_k \ \leq \ \frac{1}{c} w_j \ \leq \ x_j \ \leq \ x_k.
\end{equation}

\end{proof}

In summary, the proposed auction mechanism, OAUSA, guarantees an network throughput within $\displaystyle \frac{1 - F(x_k)}{c \cdot p(x_k)}$ units from the maximum throughput offered by its counterpart, namely VCG auction. 
\end{magentatext}


\section{Truthful Revelation of Sensing Decisions \label{sec: Truthful Sensing}}
Given that the CRs are selfish and rational, there is a motivation for the CRs to deviate from truthful revelation of sensing decisions if they obtain a greater utility. In this subsection, we investigate the conditions under which CRs falsify their sensing-decisions, and propose a novel fusion rule at the moderator that enforces truthful revelation of sensing decisions.


\begin{redtext}
In order to illustrate this rational behavior at the CRs, we consider an example network where the CRs have the same false-alarm and detection probabilities, denoted as $P_f$ and $P_d$ respectively, and the moderator employs a \emph{k-out-of-N} fusion rule. In the following lemma, we show the conditions under which a CR would unilaterally deviate from truthful revelation of its sensing decision for this example.

\begin{lemma}
Let all the CRs operate with identical false-alarm probabilities, denoted as $P_f$. Then, if the moderator employs a $k$-out-of-$N$ fusion rule, the $i^{th}$ CR unilaterally maximizes its utility when it always reveals a '0' to the moderator.
\end{lemma}
\begin{proof}
Assuming that the $i^{th}$ CR deviates unilaterally in revealing the sensing decisions to the moderator, we denote the sensing decision revealed by the $i^{th}$ CR as $\tilde{u}_i$. We assume that $\tilde{u}_i$ is constructed from $u_i$ via a stochastic flipping model, where $\alpha_1 = Pr(\tilde{u}_i = 1|u_i = 0)$ and $\alpha_2 = Pr(\tilde{u}_i = 0|u_i = 1)$. Therefore, the global false alarm probability $\tilde{Q}_F$ due to unilateral deviation of the $i^{th}$ CR is given as follows:
\begin{equation}
	\begin{array}{lcl}
		\tilde{Q}_F & = & \displaystyle \tilde{P}_f \sum_{j = k-1}^{N-1} \nchoosek{N-1}{j} P_f^j \left( 1-P_f \right)^{N-1-j}
		\\[2ex]
		&& \quad + \displaystyle (1 - \tilde{P}_f) \sum_{j = k}^{N-1} \nchoosek{N-1}{j} P_f^j \left( 1-P_f \right)^{N-1-j},
	\end{array}
\end{equation}
where $\tilde{P}_f = \alpha_1 + (1 - \alpha_1 - \alpha_2)P_F$ is the conditional probability of $\tilde{u}_i = 1$ under $H_0$. 

On simplification, we get
\begin{equation}
	\begin{array}{lcl}
		\tilde{Q}_F & = & \displaystyle \tilde{P}_f \nchoosek{N-1}{k-1} P_f^{k-1} \left( 1-P_f \right)^{N-k}
		\\[2ex]
		&& \quad + \displaystyle \sum_{j = k}^{N-1} \nchoosek{N-1}{j} P_f^j \left( 1-P_f \right)^{N-1-j}.
	\end{array}
	\label{Eqn: Qf - falsified}
\end{equation}
Note that, when $\alpha_1 = 1$ and $\alpha_2 = 0$, we have $\tilde{P}_f = 1$, which results in an upper bound on $\tilde{Q}_F$. On the other hand, when $\alpha_1 = 0$ and $\alpha_2 = 1$, we have $\tilde{P}_f = 0$. This results in a lower bound on $\tilde{Q}_F$. 

In our proposed auction, the $i^{th}$ CR's utility increases if $\tilde{Q}_f$ decreases, as shown in the following equation.
\begin{equation}
	\begin{array}{lcl}
		\tilde{U}_i(t_i) & = & \displaystyle \mathbb{E}_{\mathbf{t}_{-i}} \left[ \tilde{q}_0 \psi_i t_i - b_i - c_p \right]
		\\[2ex]
		& = & \displaystyle \tilde{q}_0 \mathbb{E}_{\mathbf{t}_{-i}} \left[ \int_{a_j}^{t_j} \psi_j(s_j, t_{-j}) d_{s_j} \right]
		\\[2ex]
		& = & \displaystyle \pi_0(1 - \tilde{Q}_f) \mathbb{E}_{\mathbf{t}_{-i}} \left[ \int_{a_j}^{t_j} \psi_j(s_j, t_{-j}) d_{s_j} \right]
	\end{array}
\end{equation}

Therefore, the $i^{th}$ CR always transmits a '0' by choosing $\alpha_1 = 0$ and $\alpha_2 = 1$ in order to attain the minimum value of $\tilde{Q}_f$ and maximize its own utility.
\end{proof}
\end{redtext}

%

Since the $i^{th}$ CR's utility can always be improved by unilaterally transmitting a '0' instead of its true sensing decision, we focus our attention on the design of a strategy-proof fusion rule that discourages the CRs from this unilateral deviation.

\begin{thrm}
	A fusion rule is strategy-proof (truthful) if the moderator disregards sensing decisions from all the CRs with $w^* = \max \mathbf{w}$ in the fusion process.
	\label{Thrm: Strategy-Proof Fusion}
\end{thrm}
\begin{proof}
	Let the $i^{th}$ CR have the maximum virtual valuation, i.e., $w_i = \max \mathbf{w}$. We assume that the moderator ignores the $i^{th}$ CR's sensing decision in its fusion rule. Then, for any CR-index $j(\neq i)$, we have $w_j < w_i$. Consequently, the moderator does not allocate the spectrum to the $j^{th}$ CR. Therefore, since $\psi_j(s_j, \mathbf{t}_{-j}) = 0$ for all $s_j \in [a_j, t_j]$ irrespective of whether the $j^{th}$ CR tries to increase $q_0$, the utility of the $j^{th}$ CR is given by
	\begin{equation}
		\begin{array}{lcl}
			\tilde{U}^{(m)}_j(t_j) & = & \displaystyle \mathbb{E}_{\mathbf{t}_{-j}} \left[ q_0 \psi_j t_j - b_j - c_p \right]
			\\
			\\
			& = & \displaystyle q_0 \mathbb{E}_{\mathbf{t}_{-j}} \left[ \int_{a_j}^{t_j} \psi_j(s_j, t_{-j}) d_{s_j} \right]
			\\
			\\
			& = & 0.
		\end{array}
	\end{equation}
	
	Now, consider the case where $j = i$. Then, since the $i^{th}$ CR's sensing decision is not considered by the moderator in its fusion rule, it cannot influence $q_0$. Therefore, no CR in the network can increase its utility by employing a miss strategy. 
	
	Furthermore, consider the event where there are multiple winners in our proposed auction mechanism. 
	If such an event takes place, based on the above mentioned arguments, it can be shown in a straightforward manner that a fusion rule is truthful if the moderator disregards all the winning CRs' sensing decisions.
\end{proof}

\begin{redtext}
Given that the sensing performance changes with the removal of CR decisions from the fusion rule, we present some analysis on the change in the moderator's utility in the following lemma.

\begin{thrm}
Let $\tilde{U}_0$ denote the moderator's utility when the sensing decisions from the CR with $w^* = \max \mathbf{w}$ are removed from the $k$-out-of-$N$ fusion rule. If $U_0$ is the moderator's utility when all the CRs' sensing decisions are accounted for, we have
\begin{equation}
	\begin{array}{lcl}
		|U_0 - \tilde{U}_0|  & \leq & \left| \pi_0 \mathbb{E}(w_{max}) P_f^N - \pi_1 c_{coll} P_d^N \right|
		\\[2ex]
		&& + \displaystyle \sum_{j = k}^{N-1} \left[ \pi_1 c_{coll} \nchoosek{N-1}{j} P_d^j (1-P_d)^{N-1-j} \right.
		\\[2ex]
		&& \displaystyle \left. + \pi_0 \mathbb{E}(w_{max}) \nchoosek{N-1}{j-1} P_f^j (1-P_f)^{N-1-j} \right].
	\end{array}
	\label{Eqn: U_0 - diff - absolute}
\end{equation} 

Furthermore, as $N \rightarrow \infty$, we have
\begin{equation}
	\displaystyle \lim_{N \rightarrow \infty} (U_0 - \tilde{U}_0) = 0.
\end{equation} 
\end{thrm}
\begin{proof}
Let $Q_F$ and $Q_D$ denote the global probabilities of false-alarm and detection at the moderator, when all the CRs' sensing decisions are considered in the $k$-out-of-$N$ fusion rule. We assume that the $i^{th}$ CR\footnote{This assumption that there is only one CR with $w^* = \max \mathbf{w}$ holds true whenever valuations are sampled from an uncountable set of real numbers since $Pr(w_i = w_j) = 0$, for any $i \neq j$} has $w^* = \max \mathbf{w}$. In other words, the proposed strategy-proof fusion rule, as proposed in Theorem \ref{Thrm: Strategy-Proof Fusion}, necessitates the removal of the $i^{th}$ CR's sensing decision. In such a case, we assume that the moderator's utility is $\tilde{U}_0$. Both $U_0$ and $\tilde{U}_0$ can be found from Equation \eqref{Eqn: U0 - final} as follows:
\begin{equation}
	\begin{array}{lcl}
		U_0 & = & q_0 \mathbb{E}(w_{max}) - q_1 c_{coll} - N c_p
		\\[2ex]
		\tilde{U}_0 & = & \tilde{q}_0 \mathbb{E}(w_{max}) - \tilde{q}_1 c_{coll} - N c_p
	\end{array}
	\label{Eqn: U_0 - Both fusion rules}
\end{equation}

Therefore, the absolute change in the moderator's utility due to the removal of the $i^{th}$ CR's sensing decision is given by
\begin{equation}
	\begin{array}{lcl}
		U_0 - \tilde{U}_0  & = & (q_0 - \tilde{q}_0) \mathbb{E}(w_{max}) - (q_1 - \tilde{q}_1) c_{coll},
	\end{array}
	\label{Eqn: U_0 - diff}
\end{equation}
where $q_0 - \tilde{q}_0 = \pi_0 ( Q_F - \tilde{Q}_F )$ and $q_1 - \tilde{q}_1 = \pi_1 ( Q_D - \tilde{Q}_D )$.

From Equation \eqref{Eqn: Qf - falsified}, we have
\begin{equation}
	\begin{array}{lcl}
		Q_F - \tilde{Q}_F & = & \displaystyle \sum_{j = k}^N \nchoosek{N}{j} P_f^j (1 - P_f)^{N-j}
		\\
		&& \displaystyle - \sum_{j = k}^{N-1} \nchoosek{N-1}{j} P_f^j (1 - P_f)^{N-1-j}
		\\[3ex]
		& \leq & \displaystyle P_f^N + \sum_{j = k}^{N-1} \nchoosek{N-1}{j - 1} P_f^j (1 - P_f)^{N-1-j}.
	\end{array}
	\label{Eqn: Qf - diff - falsified}
\end{equation}
since $\displaystyle \nchoosek{N}{j} - \nchoosek{N-1}{j} = \nchoosek{N-1}{j-1}$ and $(1 - P_f) \leq 1$.

Using a similar approach, we can calculate a lower bound on $Q_D - \tilde{Q}_D$ as follows:
\begin{equation}
	\begin{array}{lcl}
		Q_D - \tilde{Q}_D & = & \displaystyle \sum_{j = k}^N \nchoosek{N}{j} P_d^j (1 - P_d)^{N-j}
		\\
		&& \displaystyle - \sum_{j = k}^{N-1} \nchoosek{N-1}{j} P_d^j (1 - P_d)^{N-1-j}.
	\end{array}
	\label{Eqn: Qd - diff - falsified}
\end{equation}

Replacing the term $\nchoosek{N}{j}$ with a smaller term $\nchoosek{N-1}{j}$, we have
\begin{equation}
	\begin{array}{lcl}
		Q_D - \tilde{Q}_D & \geq & \displaystyle P_d^N - \sum_{j = k}^{N-1} \nchoosek{N-1}{j} P_d^j (1 - P_d)^{N-1-j}.
	\end{array}
	\label{Eqn: Qd - diff - falsified}
\end{equation}

Substituting Equations \eqref{Eqn: Qf - diff - falsified} and \eqref{Eqn: Qd - diff - falsified} in Equation \eqref{Eqn: U_0 - diff}, we have
\begin{equation}
	\begin{array}{lcl}
		U_0 - \tilde{U}_0  & \leq & \left( \pi_0 \mathbb{E}(w_{max}) P_f^N - \pi_1 c_{coll} P_d^N \right)
		\\[2ex]
		&& + \displaystyle \sum_{j = k}^{N-1} \left[ \pi_1 c_{coll} \nchoosek{N-1}{j} P_d^j (1-P_d)^{N-1-j} \right.
		\\[2ex]
		&& \displaystyle \left. + \pi_0 \mathbb{E}(w_{max}) \nchoosek{N-1}{j-1} P_f^j (1-P_f)^{N-1-j} \right]
	\end{array}
	\label{Eqn: U_0 - diff - final}
\end{equation}

Since $|a \pm b| \leq |a| + |b|$ for any real values of $a$ and $b$, we get Equation \eqref{Eqn: U_0 - diff - absolute} from Equation \eqref{Eqn: U_0 - diff - final}.

Next, we investigate the asymptotic behavior of $U_0 - \tilde{U}_0$ as $N$ tends to infinity. It is well-known \cite{Book-Varshney} that, for any positive integral values of $n$ and $k$ such that $1 \leq k \leq n$, we have
\begin{equation}
	\begin{array}{lclclcl}
		\displaystyle \lim_{N \rightarrow \infty} Q_F & = & 0 & \quad & \displaystyle \lim_{N \rightarrow \infty} \tilde{Q}_F & = & 0
		\\
		\displaystyle \lim_{N \rightarrow \infty} Q_D & = & 1 & \quad & \displaystyle \lim_{N \rightarrow \infty} \tilde{Q}_D & = & 1
	\end{array}
\end{equation}

Since $\displaystyle \lim_{N \rightarrow \infty} (Q_F - \tilde{Q}_F) = 0$ and $\displaystyle \lim_{N \rightarrow \infty} (Q_D - \tilde{Q}_D) = 0$, we have $\displaystyle \lim_{N \rightarrow \infty} (U_0 - \tilde{U}_0) = 0$.
\end{proof}
\end{redtext}


Note that, since the valuations are independent of sensing decisions, the choice of the fusion rule has no impact on the truthfulness of the auction mechanism in terms of revelation of CR valuations. Therefore, our proposed auction is strategy-proof in multiple dimensions, namely truthful revelation of sensing decisions and valuations.

\section{Simulation Results \label{sec: Simulation}}
\begin{figure}[!t]
	\centering
    \includegraphics[width=3.3in]{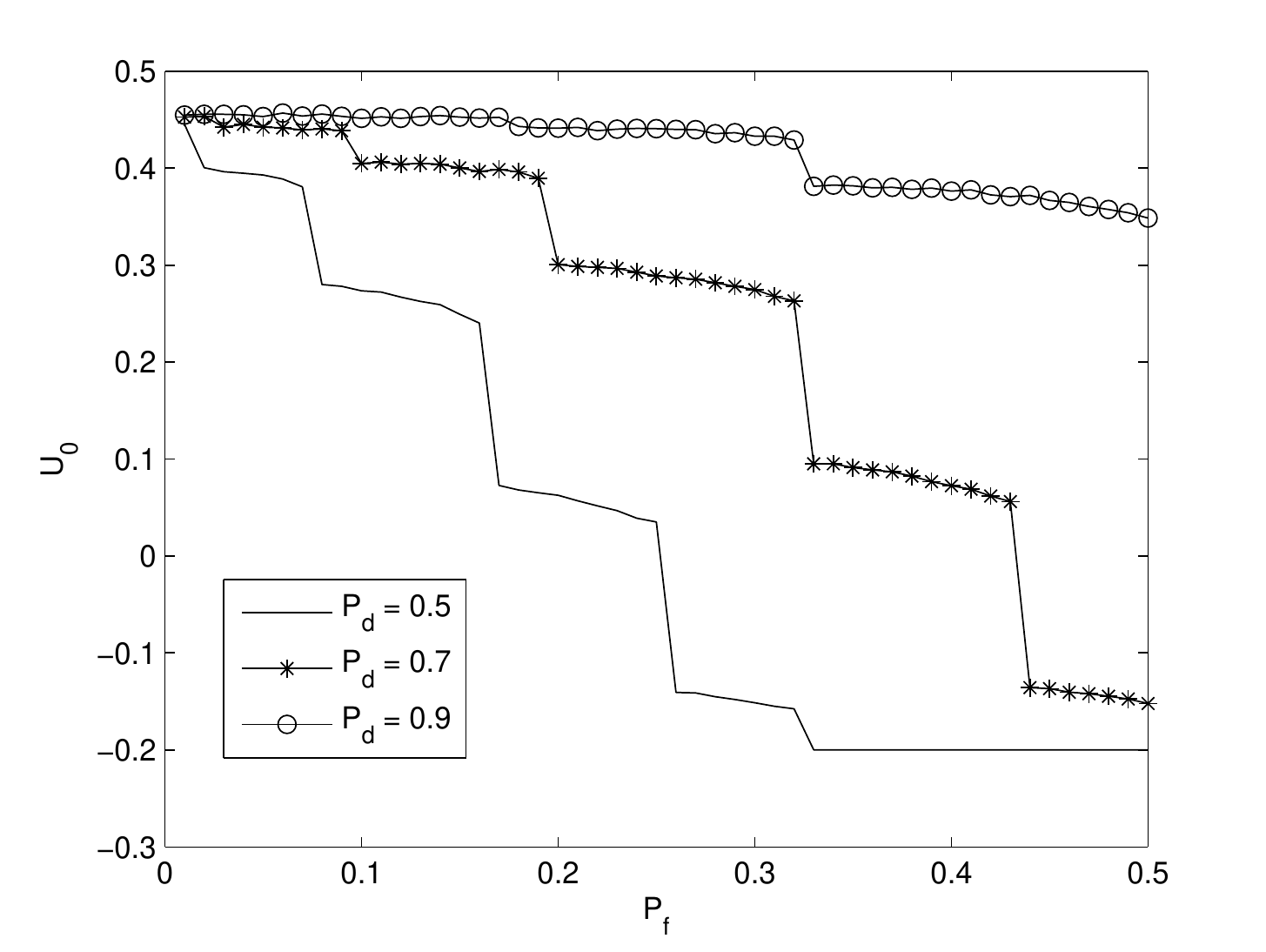}
    \caption{Utility of the moderator vs. local false alarm probability, for a fixed $P_d$.}
    \label{Fig: U_mod_vs_Pf}
\end{figure}

In our simulation results, we consider a CR network with 10 identical nodes, where we assume the prior distribution of the PU spectrum usage as $\pi_0 = 0.8$ (or equivalently, $\pi_1 = 0.2$) \cite{Akyildiz2006}. Also, let the probabilities of false-alarm and detection be denoted as $P_{f_i} = P_f$ and $P_{d_i} = P_d$. We also assume that the moderator employs a \emph{k-out-of-N} fusion rule in order to make a global inference about PU's spectrum availability, where $k$ is chosen to be optimal in the Bayesian sense \cite{Nadendla-Thesis}, as follows.
\begin{equation}
	k_{opt} = \displaystyle \left\lceil \frac{\displaystyle \log \left( \frac{\pi_1}{\pi_0} \right) + N \log\left( \frac{1-P_f}{1-P_d} \right)}{\displaystyle \log\left( \frac{P_d (1 - P_f)}{P_f (1 - P_d)}\right)} \right\rceil.
	\label{Eqn: Optimal k}
\end{equation}

Note that $k_{opt}$ is chosen to minimize the error probability at the moderator. It depends on the number of CRs and their probabilities of false alarm and detection, along with the prior probabilities of the hypotheses. Therefore, if any of these parameters change, $k_{opt}$ changes in steps due to the ceil function ($\lceil \cdot \rceil$).

Then, the global probabilities of false alarm and detection, denoted $Q_f$ and $Q_d$ respectively, can be found as follows.
\begin{subequations}
	\begin{gather}
		\label{eqn3a}
		Q_f = \sum_{i = k_{opt}}^{N} \nchoosek{N}{i} P_f^i \left( 1-P_f \right)^{N-i},
        \\
        \label{eqn3b}
        Q_d = \sum_{i = k_{opt}}^{N} \nchoosek{N}{i} P_d^i \left( 1-P_d \right)^{N-i}.
	\end{gather}
	\label{Eqn: CR sensing performance}
\end{subequations}

\begin{figure}[!t]
	\centering
    \includegraphics[width=3.3in]{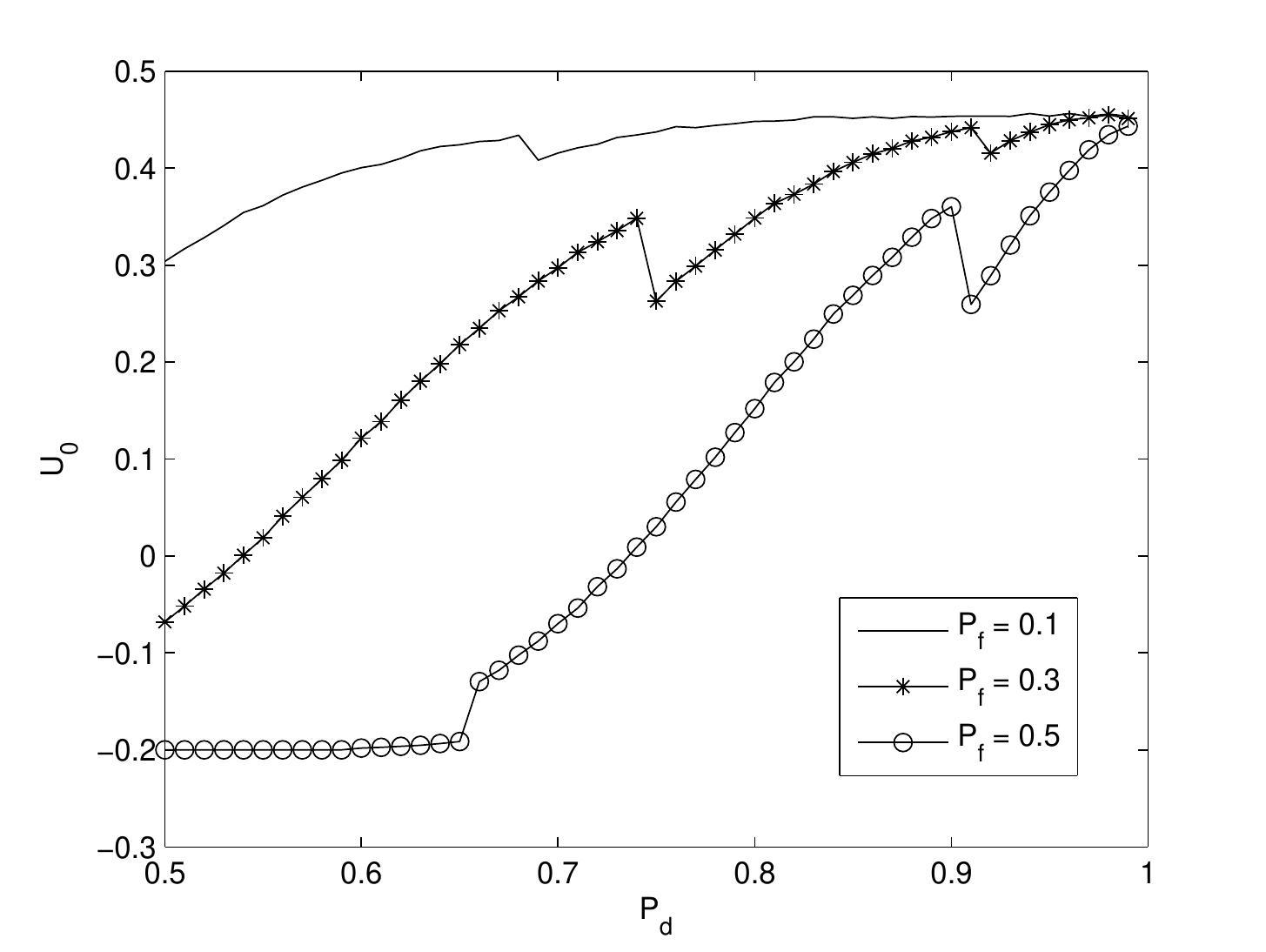}
    \caption{Utility of the moderator vs. local probability of detection, for a fixed $P_f$.}
    \label{Fig: U_mod_vs_Pd}
\end{figure}

We also assume that the valuation $t_i$ of $i^{th}$ CR is uniformly distributed over [0,1], i.e.  $\mathcal{U}[0,1], \ \forall i = 1, \cdots, N$. Also, since both the optimal allocation and CRs' payments depend on $\mathbf{t}$, we obtain the results for the proposed auction based on 10000 Monte-Carlo runs.

First, we investigate the behavior of the proposed auction in terms of the false-alarm and detection probabilities, $P_f$ and $P_d$ at the spectrum sensors embedded within the CRs. In order to illustrate this behavior, we assume $c_p = 0.02$ and $c_{coll} = 5$ in our simulation results. In Figure \ref{Fig: U_mod_vs_Pf}, we plot the expected utility of the moderator in terms of $P_f$, when the CRs' detection probability is fixed at $P_d = 0.5, 0.7, 0.9$. As per our intuition, the moderator's expected utility decreases with increasing $P_f$, since the moderator misses to detect available opportunities even though the PU channel is idle. Note that the staircase pattern is due to the change in the optimal value of $k$, which is given in Equation (\ref{Eqn: Optimal k}). Similarly, in Figure \ref{Fig: U_mod_vs_Pd}, we plot the variation of the moderator's expected utility with respect to the CRs' detection probability $P_d$, for different values of $P_f$. Note that the plot resembles an improving sawtooth curve, where the sudden drop in moderator's utility occurs whenever the value of the optimal $k$ decreases with increasing $P_d$. In other words, if $k$ decreases by one unit, a new term appears in $Q_d$, as given in Equation (\ref{Eqn: CR sensing performance}). Since $q_0$ is a decreasing function of $Q_d$, the expected utility of the moderator abruptly drops and then improves with increasing $P_d$ as long as the value of $k$ remains fixed. Also, note that the moderator's utility becomes negative whenever the CRs have very low sensing performance (low $P_d$, and high $P_f$) in both Figures \ref{Fig: U_mod_vs_Pf} and \ref{Fig: U_mod_vs_Pd}. This behavior can be explained by the feasibility condition presented in Theorem \ref{Thrm: Feasibility}.

\begin{figure}[!t]
	\centering
    \includegraphics[width=3.3in]{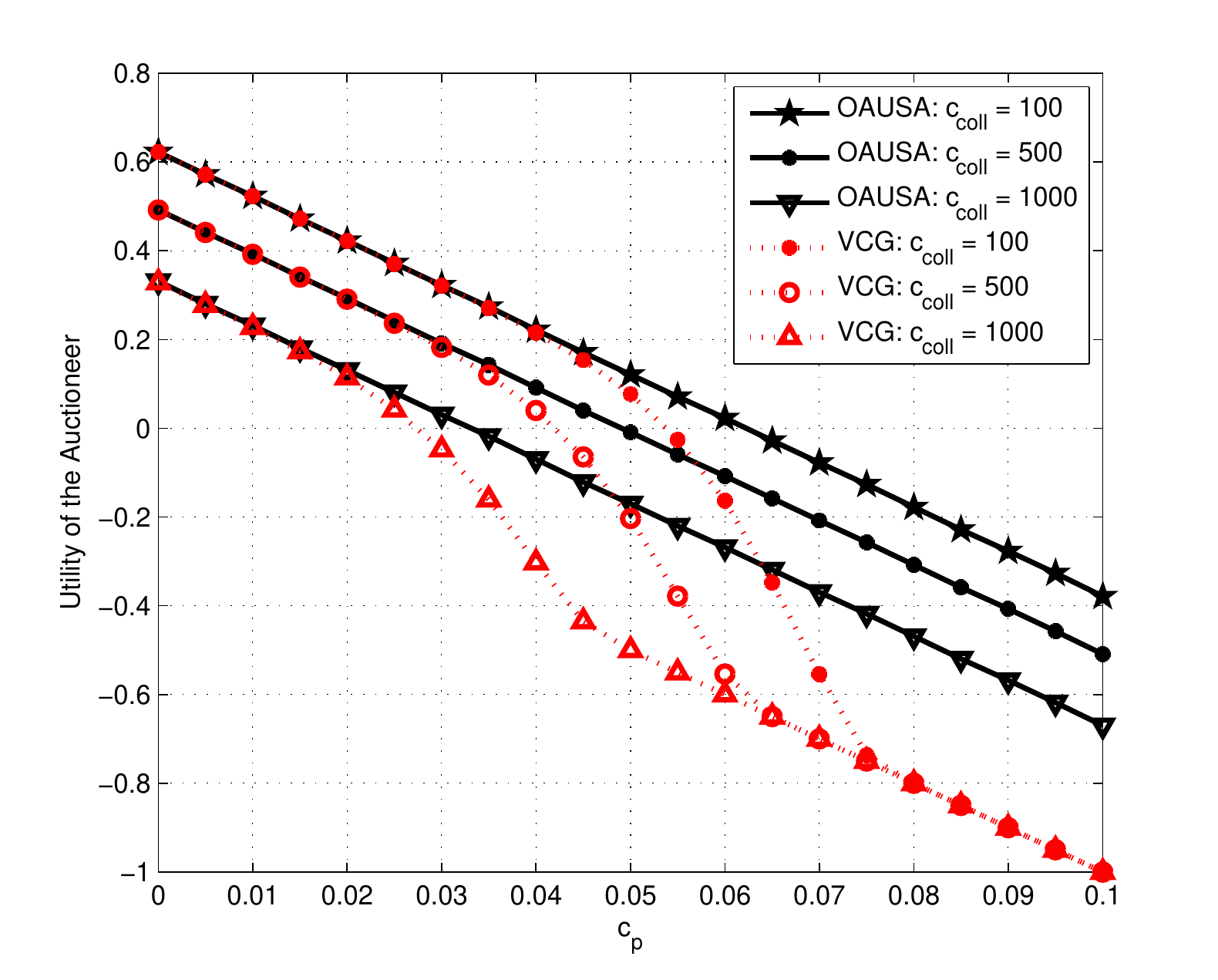}
    \caption{Utility of the moderator vs. participation cost at the CRs}
    \label{Fig: U_mod_vs_c_p}
\end{figure}

In our next set of simulations, we fix the probabilities of false alarm and detection at the CRs as $P_f = 0.1$ and $P_d = 0.9$. In Figure \ref{Fig: U_mod_vs_c_p}, we present the variation of the moderator's utility $U_0$, with respect to the sensing cost $c_p$. In our simulation, we fix the cost of collision to $c_{coll} = 100, 500, 1000$, and vary the cost of sensing $c_p$ over a range of $(0,0.1)$. 
Similarly, in Figure \ref{Fig: U_mod_vs_c_coll}, we present simulation results of the moderator's utility for the proposed auction in terms of the cost of collision $c_{coll}$. In our results, we fix the participation costs to $c_p = 0.01, 0.02, 0.05$. As expected, we observe that the moderator's utility decreases linearly with increase in $c_{coll}$. Similar to our simulation results in Figures \ref{Fig: U_mod_vs_Pf} and \ref{Fig: U_mod_vs_Pd}, as shown in Theorem \ref{Thrm: Feasibility}, we observe that the proposed auction becomes infeasible for high $c_p$ and $c_{coll}$. Also, the optimal auction always provides a higher utility at the moderator, when compared to the modified VCG auction as shown in Figures \ref{Fig: U_mod_vs_c_p} and \ref{Fig: U_mod_vs_c_coll}.

\begin{figure}[!t]
	\centering
    \includegraphics[width=3.3in]{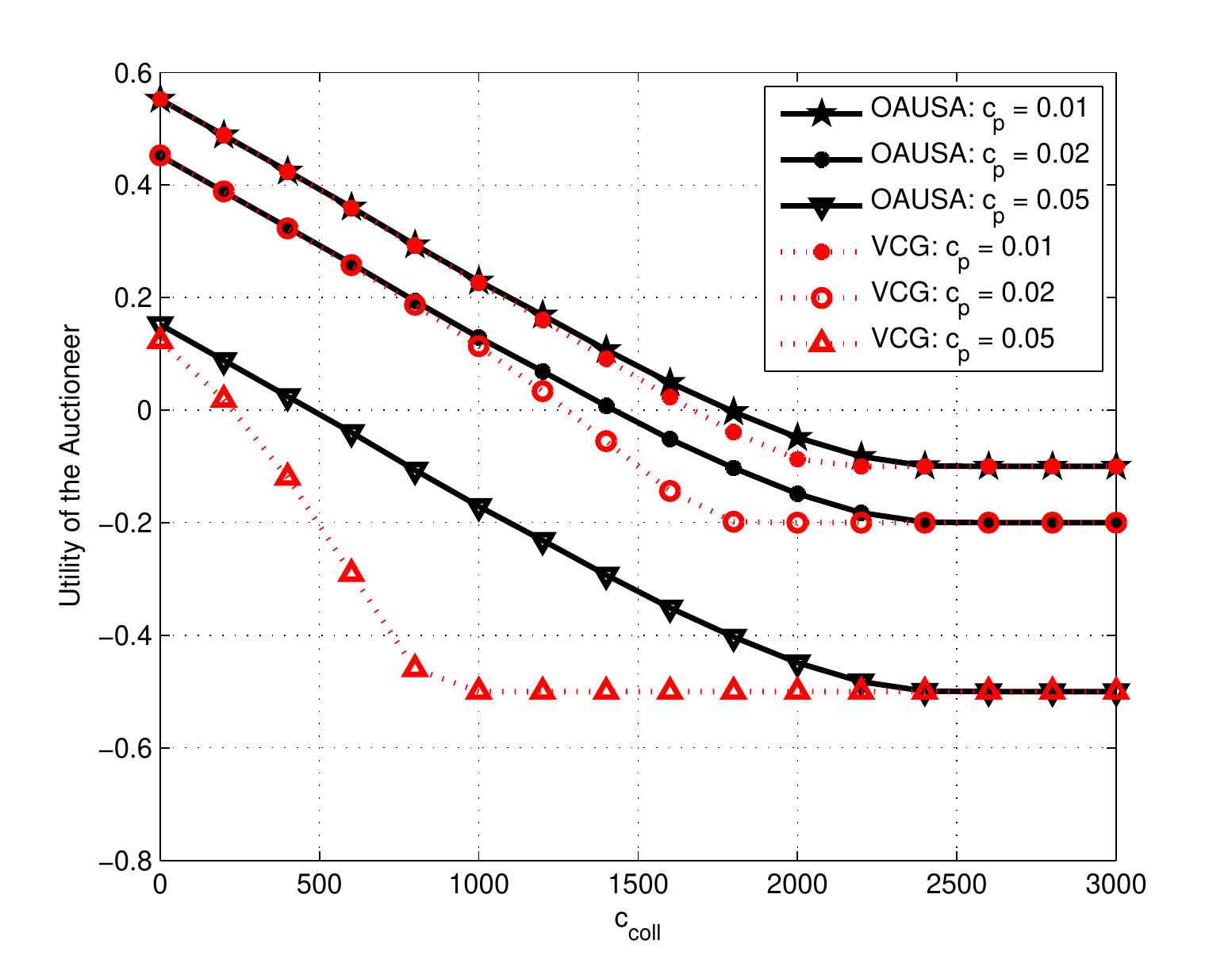}
    \caption{Utility of the moderator vs. collision penalty}
    \label{Fig: U_mod_vs_c_coll}
\end{figure}

\section{Conclusion and Future Directions \label{sec: Conclusion}}
In summary, we have designed the optimal auction for CR networks in the presence of uncertainty in the availability of PU spectrum at the moderator. We have considered cooperative spectrum sensing in our system model to mitigate the effects of uncertainty at the moderator and improve efficiency in spectrum utilization. Due to the presence of participation and collision costs, in addition to the spectrum uncertainty at the moderator, we have investigated necessary conditions under which our proposed auction is feasible. Furthermore, we have addressed the issue to truthful revelation of sensing decisions, making our auction mechanism strategy-proof in multiple dimensions. We have compared the performance of our proposed auction with traditional auctions such as VCG auction in terms of expected utility at the moderator and network throughput. Numerical results were also provided to illustrate the performance of the proposed auction under different scenarios. In our future work, we will investigate the problem of spectrum markets under uncertain spectrum availability where multiple PUs are present in the network. In addition, we will also investigate the case where a given spectrum can be spatially reallocated to different CRs in an interference aware manner.

\bibliographystyle{IEEEtran}
\bibliography{references}

\end{document}